\documentclass[a4paper,UKenglish,cleveref, autoref, thm-restate,authorcolumns]{lipics-v2019}

\usepackage{mathtools}
\usepackage{amsfonts}
\usepackage{comment}
\usepackage{algorithm}
\usepackage[noend]{algpseudocode}
\usepackage{xspace}
\usepackage[colorinlistoftodos,textsize=small,color=blue!25!white,obeyFinal]{todonotes}

\usepackage{caption}
\usepackage{ulem}
\usepackage{multirow}
\usepackage{subcaption}
\usepackage{ifthen}
\DeclareMathOperator{\dist}{dist}
\allowdisplaybreaks
\nolinenumbers
\title{Multi-level Weighted Additive Spanners\footnote{This paper has been accepted in the 19th Symposium on Experimental Algorithms, SEA 2021.}} 

\author{Reyan~Ahmed}{University of Arizona, Tucson, United States}{abureyanahmed@email.arizona.edu}{}{}
\author{Greg~Bodwin}{University of Michigan, Ann Arbor, United States}{bodwin@umich.edu}{}{}
\author{Faryad~Darabi~Sahneh}{University of Arizona, Tucson, United States}{faryad@email.arizona.edu}{}{}
\author{Keaton~Hamm}{University of Texas at Arlington, Arlington, United States}{keaton.hamm@uta.edu}{}{}
\author{Stephen~Kobourov}{University of Arizona, Tucson, United States}{kobourov@cs.arizona.edu}{}{}
\author{Richard~Spence}{University of Arizona, Tucson, United States}{rcspence@email.arizona.edu}{}{}

\authorrunning{R.\ Ahmed, et al.} 

\Copyright{Reyan Ahmed, Greg Bodwin, Faryad Darabi Sahneh, Keaton Hamm, Stephen Kobourov, and Richard Spence} 



\category{} 

\relatedversion{} 

\supplement{All algorithms, implementations, the ILP solver, experimental data and analysis are available on Github at \url{ https://github.com/abureyanahmed/multi_level_weighted_additive_spanners}.
}

\acknowledgements{The research for this paper was partially supported by NSF grants CCF-1740858, CCF1712119, and DMS-1839274.}

\newcommand{\OPT}{\text{OPT}}

\newcommand{\Cc}{\mathcal{C}}
\newcommand{\eps}{\varepsilon}

\begin{document}

\maketitle

\begin{abstract}
Given a graph $G = (V,E)$, a subgraph $H$ is an \emph{additive $+\beta$ spanner} if $\dist_H(u,v) \le \dist_G(u,v) + \beta$ for all $u, v \in V$. A \emph{pairwise spanner} is a spanner for which the above inequality only must hold for specific pairs $P \subseteq V \times V$ given on input, and when the pairs have the structure $P = S \times S$ for some subset $S \subseteq V$, it is specifically called a \emph{subsetwise spanner}.
Spanners in unweighted graphs have been studied extensively in the literature, but have only recently been generalized to weighted graphs.

In this paper, we consider a multi-level version of the subsetwise spanner in weighted graphs, where the vertices in $S$ possess varying level, priority, or quality of service (QoS) requirements, and the goal is to compute a nested sequence of spanners with the minimum total number of edges. We first generalize the $+2$ subsetwise spanner of [Pettie 2008, Cygan et al., 2013] to the weighted setting.
We experimentally measure the performance of this and several other algorithms for weighted additive spanners, both in terms of runtime and sparsity of output spanner, when applied at each level of the multi-level problem.
Spanner sparsity is compared to the sparsest possible spanner satisfying the given error budget, obtained using an integer programming formulation of the problem.
We run our experiments with respect to input graphs generated by several different random graph generators: Erd\H{o}s--R\'{e}nyi, Watts--Strogatz, Barab\'{a}si--Albert, and random geometric models. By analyzing our experimental results we developed a new technique of changing an initialization parameter value that provides better performance in practice.

\end{abstract}

\clearpage
\setcounter{page}{1}


\newpage
\section{Introduction}

This paper studies spanners of undirected input graphs with edge weights.
Given an input graph, a \emph{spanner} is a sparse subgraph with approximately the same distance metric as the original graph.
Spanners are used as a primitive for many algorithmic tasks involving the analysis of distances or shortest paths in enormous input graphs; it is often advantageous to first replace the graph with a spanner, which can be analyzed much more quickly and stored in much smaller space, at the price of a small amount of error.
See the recent survey \cite{ahmed2020graphElsevier} for more details on these applications.

Spanners were first studied in the setting of \emph{multiplicative error}, where for an input graph $G$ and an error (``stretch'') parameter $k$, the spanner $H$ must satisfy $\dist_H(s, t) \le k \cdot \dist_G(s, t)$ for all vertices $s, t$.
This setting was quickly resolved in a seminal paper by Alth{\" o}fer, Das, Dobkin, Joseph, and Soares \cite{alth90}, where the authors proved that for all positive integers $k$, all $n$-vertex graphs have spanners on $O(n^{1+1/k})$ edges with stretch $2k-1$, and that this tradeoff is best possible.
Thus, as expected, one can trade off error for spanner sparsity, increasing the stretch $k$ to pay more and more error for sparser and sparser spanners.

For very large graphs, \emph{purely additive error} is arguably a much more appealing paradigm. For a constant $c > 0$, a \emph{$+c$ spanner} of an $n$-vertex graph $G$ is a subgraph $H$ such that $\dist_H(s, t) \le \dist_G(s, t)+c$
for all vertices $s,t$.
Thus, for additive error the excess distance in $H$ is \emph{independent} of the graph size and of $\dist_G(s, t)$, which can be large when $n$ is large.
Additive spanners were introduced by Liestman and Shermer~\cite{liestmannp},
and followed by three landmark theoretical results on the sparsity of additive spanners in unweighted graphs: Aingworth, Chekuri, Indyk, and Motwani~\cite{Aingworth99fast} showed that all graphs have $+2$ spanners on $O(n^{3/2})$ edges, Chechik~\cite{chechik2013new, bodwin2020note} showed that all graphs have $+4$ spanners on $O(n^{7/5})$ edges, and Baswana, Kavitha, Mehlhorn, and Pettie~\cite{baswana2010additive} showed that all graphs have $+6$ spanners on $O(n^{4/3})$ edges.

Despite the inherent appeal of additive error, spanners with multiplicative error remain much more commonly used in practice.
There are two reasons for this.
\begin{enumerate}
    \item First, while the multiplicative spanner of Alth{\" o}fer et al~\cite{alth90} works without issue for weighted graphs, the previous additive spanner constructions hold only for unweighted graphs, whereas the metrics that arise in applications often require edge weights.
    Addressing this, recent work of the authors \cite{ahmed2020weighted} and in two papers of Elkin, Gitlitz, and Neiman \cite{elkin2019almost, elkin2020improved} gave natural extensions of the classic additive spanner constructions to weighted graphs.
    For example, the $+2$ spanner bound becomes the following statement: for all $n$-vertex weighted graphs $G$, there is a subgraph $H$ satisfying $\dist_H(s, t) \le \dist_G(s, t) + 2W(s, t)$, where $W(s, t)$ denotes the maximum edge weight along an arbitrary $s \leadsto t$ shortest path in $G$.
    The $+4$ spanner generalizes similarly, and the $+6$ spanner does as well with the small exception that the error increases to $+(6+\eps)W(s, t)$, for $\eps > 0$ arbitrarily small which trades off with the implicit constant in the spanner size.

    \item Second, $\text{poly}(n)$ factors in spanner size can be quite serious in large graphs, and so applications often require spanners of near-linear size, say $O(n^{1.01})$ edges for an $n$-vertex input graph.
    The \emph{worst-case} spanner sizes of $O(n^{4/3})$ or greater for additive spanner constructions are thus undesirable, and unfortunately, there is a theoretical barrier to improving them: Abboud and Bodwin \cite{abboud20174} proved that one \emph{cannot} generally trade off more additive error for sparser spanners, as one can in the multiplicative setting.
    Specifically, for any constant $c > 0$, there is no general construction of $+c$ spanners for $n$-node input graphs on $O(n^{4/3 - 0.001})$ edges.
    However, the lower bound construction is rather pathological, and it is not likely to arise in practice.
    It is known that for many practical graph classes, e.g., those with good expansion, near-linear additive spanners always exist \cite{baswana2010additive}.
    Thus, towards applications of additive error, it is currently an important open question whether modern additive spanner constructions on \emph{practical} graphs of interest tend to exhibit performance closer to the worst-case bounds from \cite{abboud20174}, or bounds closer to the best ones available for the given input graphs.
    (We remark here that there are strong computational barriers to designing algorithms that achieve the sparsest possible $+c$ spanners directly, or which even closely approximate this quantity in general \cite{chlamtavc2017approximating}).
\end{enumerate}

The goal of this work is to address the second point, by measuring the experimental performance of the state-of-the-art constructions for weighted additive spanners on graphs generated from various popular random models and measuring their performance.
We consider both $+cW$ spanners (where $W=\max_{uv \in E}w(uv)$ is the maximum edge weight) and $+cW(\cdot, \cdot)$ spanners, whose additive error is $+cW(s,t)$ for each pair $s,t \in V$. We are interested both in runtime and in the ratio of output spanner size to the size of the sparsest possible spanner (which we obtain using an ILP, discussed in Section \ref{sec:exact_algo}).
We specifically consider generalizations of the three staple constructions for weighted additive spanners mentioned above, in which the spanner distance constraint only needs to be satisfied for given pairs of vertices.

In particular, the following extensions are considered.
A \emph{pairwise spanner} is a subgraph that must satisfy the spanner error inequality for a given set of vertex pairs $P$ taken on input, and a \emph{subsetwise spanner} is a pairwise spanner with the additional structure $P = S \times S$ for some vertex subset $S$.
See~\cite{Pettie09, Cygan13, Kavitha15, Kavitha2017, bodwin2017linear, coppersmith2006sparse, Bodwin:2016:BDP:2884435.2884496, Hsien2018nearoptimal} for recent prior work on pairwise and subsetwise spanners.
We also discuss a multi-level version of the subsetwise additive spanner problem where we have an edge-weighted graph $G=(V, E)$, a nested sequence of terminals $S_\ell \subseteq S_{\ell-1} \subseteq \cdots \subseteq S_1 \subseteq V$ and a real number $c \ge 0$ as input. We want to compute a nested sequence of subgraphs $G_\ell \subseteq G_{\ell-1} \subseteq \cdots \subseteq G_1$ such that $G_i$ is a $+cW$ subsetwise spanner of $G$ over $S_i$. The objective is to minimize the total number of edges in all subgraphs. Similar generalizations have been studied for the Steiner tree problem under various names including Multi-level Network Design~\cite{Balakrishnan1994}, Quality of Service Multicast Tree (QoSMT)~\cite{Charikar2004ToN, Karpinski2005}, Priority Steiner Tree~\cite{Chuzhoy2008}, Multi-Tier Tree~\cite{mirchandani1996MTT}, and Multi-level Steiner Tree~\cite{MLST2018, ahmed2020kruskal}. However, multi-level or QoS generalizations of spanner problems appear to have been much less studied in literature. Section~\ref{sec:subsetwise} generalizes the $+2$ subsetwise construction~\cite{Cygan13}, and Section~\ref{sec:multi_level} generalizes to the multi-level setting.

\section{Subsetwise spanners}\label{sec:subsetwise}
All unweighted graphs have polynomially constructible $+2$ subsetwise spanners over $S \subseteq V$ on $O(n \sqrt{|S|})$ edges~\cite{Pettie09, Cygan13}. For weighted graphs, Ahmed et al.~\cite{ahmed2020weighted} recently give a $+4W$ subsetwise spanner construction, also using $O(n \sqrt{|S|})$ edges. In this section we show how to generalize the $+2$ subsetwise construction~\cite{Pettie09, Cygan13} to the weighted setting by giving a construction which produces a subsetwise $+2W$ spanner of a weighted graph (with integer edge weights in $[1,W]$) on $O(nW\sqrt{|S|})$ edges. Due to space, we omit most proof details here but describe the construction instead.

A \emph{clustering} $\Cc = \{C_1, C_2, \ldots, C_q\}$ is a set of disjoint subsets of vertices. Initially, every vertex is unclustered. The subsetwise $+2W$ construction has two steps: the clustering phase and the path buying phase. The clustering phase is exactly the same as that of~\cite{Pettie09, Cygan13} in which we construct a cluster subgraph $G_{\Cc}$ as follows: set $\beta = \log_n \sqrt{|S|W}$, and while there is a vertex $v$ with at least $\lceil{n^\beta}\rceil$ unclustered neighbors, we add a cluster $C$ to $\Cc$ containing exactly $\lceil n^\beta \rceil$ unclustered neighbors of $v$ (note that $v \not\in C$). We add to $G_{\Cc}$ all edges $vx$ ($x \in C$) and $xy$ ($x, y \in C$). When there are no more vertices with at least $\lceil{n^\beta}\rceil$ unclustered neighbors, we add all the unclustered vertices and their incident edges to $G_{\Cc}$.

In the second (path-buying) phase, we start with a clustering $\Cc$ and a cluster subgraph $G_0 := G_{\Cc}$. There are $z := \binom{|S|}{2}$ unordered pairs of vertices in $S$; let $\pi_1$, $\pi_2$, \ldots, $\pi_z$ denote the shortest paths between these vertex pairs where $\pi_i = \pi(u_i, v_i)$ has endpoints $\{u_i, v_i\}$. As in~\cite{Cygan13}, we iterate from $i=1$ to $i=z$ and determine whether to add path $\pi_i$ to the spanner. Define the cost and value of a path $\pi_i$ as follows:

\begin{align*}
\text{cost}(\pi_i) &:= \# \text{ edges of } \pi_i \text{ which are absent in } G_{i-1}\\
\text{value}(\pi_i) &:= \# \text{ pairs } (x,C) \text{ where } x \in \{u_i, v_i\}, C \in \Cc,\\
& \text{$C$ contains at least one vertex in $\pi_i$,} \\
& \text{ and } \dist_{\pi_i}(x,C) < \dist_{G_{i-1}}(x,C)
\end{align*}
If $\text{cost}(\pi_i) \le (2W+1)\text{value}(\pi_i)$, then we add (``buy'') $\pi_i$ to the spanner by letting $G_i = G_{i-1} \cup \pi_i$. Otherwise, we do not add $\pi_i$, and let $G_i = G_{i-1}$. The final spanner returned is $H = G_z$.

\begin{lemma}
For any $u_i, v_i \in  S$, we have $\dist_{H}(u_i, v_i) \leq \dist_G(u_i, v_i) + 2W$.
\end{lemma}
\begin{proof}[Proof sketch.]
The proof is largely the same as in~\cite{Cygan13} except with one main difference: in the unweighted case, the distance between any two points within the same cluster is at most 2, and it is shown in~\cite{Cygan13} that using this property, if there are $t$ edges in $\pi(u,v)$ missing from $G_{\Cc}$, then there are at least $\frac{t}{2}$ clusters containing at least one vertex on $\pi(u,v)$. In the weighted case, assuming $W$ is constant, there are $\Omega(t)$ clusters of $\Cc$ which contain at least one vertex on $\pi(u,v)$.
\end{proof}

\begin{corollary}
Let $G$ be a weighted graph with integer edge weights in $[1,W]$. Then $G$ has a $+6W$ pairwise spanner on $O(Wn|P|^{1/4})$ edges.
\end{corollary}
This follows from applying the $+8W$ construction of Ahmed et al.~\cite{ahmed2020weighted} (Appendix~\ref{section:pairwise}, Algorithm~\ref{alg:8W-pairwise}), except we use the above $+2W$ subsetwise spanner instead of the $+4W$ subsetwise spanner construction given in~\cite{ahmed2020weighted} as a subroutine. 

\section{Pairwise spanner constructions~\cite{ahmed2020weighted}} \label{section:pairwise}
Here, we provide pseudocode (Algorithms~\ref{alg:2W-pairwise}--\ref{alg:8W-pairwise}) describing the $+2W$, $+4W$, and $+8W$ pairwise spanner constructions\footnote{Using a tighter analysis or the above $+2W$ subsetwise construction in place of the $+4W$ construction in Algorithm~\ref{alg:8W-pairwise}, the additive error can be improved to $+2W(\cdot, \cdot)$, $+4W(\cdot, \cdot)$, and $+6W$ for integer edge weights.} by Ahmed et al.~\cite{ahmed2020weighted}. These spanner constructions have a similar theme: first, construct a $d$-\emph{light initialization}, which is a subgraph $H$ obtained by adding the $d$ lightest edges incident to each vertex (or all edges if the degree is at most $d$). Then for each pair $(s,t) \in P$, consider the number of edges in $\pi(s,t)$ which are absent in the current subgraph $H$. Add $\pi(s,t)$ to $H$ if the number of missing edges is at most some threshold $\ell$, or otherwise randomly sample vertices and either add a shortest path tree rooted at these vertices, or construct a subsetwise spanner among them.

\begin{algorithm}[h!]

\caption{$+2W$ pairwise spanner~\cite{ahmed2020weighted}}\label{alg:2W-pairwise}
\begin{algorithmic}[1]
\State $d = |P|^{1/3}$, $\ell = n/|P|^{2/3}$
\State $H = d$-light initialization
\State let $m'$ be the number of missing edges needed for a valid construction
\While{$m'>nd$}
\For{$(s,t) \in P$}
\State $x = |E(\pi(s,t)) \setminus E(H)|$
\If{$x \le \ell$}
\State add $\pi(s,t)$ to $H$
\EndIf
\EndFor
\State $R = $ random sample of vertices, each with probability $1/(\ell d)$
\For{$r \in R$}
\State add a shortest path tree rooted at $r$ to each vertex
\EndFor
\EndWhile
\State add the $m'$ missing edges
\State \Return{$H$}
\end{algorithmic}
\end{algorithm}

\begin{algorithm}[h!]

\caption{$+4W$ pairwise spanner~\cite{ahmed2020weighted}}\label{alg:4W-pairwise}
\begin{algorithmic}[1]
\State $d = |P|^{2/7}$, $\ell = n/|P|^{5/7}$
\State $H = d$-light initialization
\State let $m'$ be the number of missing edges needed for a valid construction
\While{$m'>nd$}
\For{$(s,t) \in P$}
\State $x = |E(\pi(s,t)) \setminus E(H)|$
\If{$x \le \ell$}
\State add $\pi(s,t)$ to $H$
\ElsIf{$x \ge n/d^2$}
\State $R_1$ = random sample of vertices, each w.p. $d^2/n$
\State add a shortest path tree rooted at each $r \in R_1$
\Else
\State add first $\ell$ and last $\ell$ missing edges of $\pi(s,t)$ to $H$
\State $R_2$ = i.i.d. sample of vertices, w.p. $1/(\ell d)$
\For{each $r, r' \in R_2$}
\If{exists $r \to r'$ path missing $\le n/d^2$ edges}
\State add to $H$ a shortest $r \to r'$ path among paths missing $\le n/d^2$ edges
\EndIf
\EndFor
\EndIf
\EndFor
\EndWhile
\State add the $m'$ missing edges
\State \Return{$H$}
\end{algorithmic}
\end{algorithm}

\begin{algorithm}[h!]

\caption{$+8W$ pairwise spanner~\cite{ahmed2020weighted}}\label{alg:8W-pairwise}
\begin{algorithmic}[1]
\State $d = |P|^{1/4}$, $\ell = n/|P|^{3/4}$
\State $H = d$-light initialization
\State let $m'$ be the number of missing edges needed for a valid construction
\While{$m'>nd$}
\For{$(s,t) \in P$}
\State $x = |E(\pi(s,t)) \setminus E(H)|$
\If{$x \le \ell$}
\State add $\pi(s,t)$ to $H$
\Else
\State add first $\ell$ and last $\ell$ missing edges of $\pi(s,t)$ to $H$
\State $R$ = random sample of vertices, each w.p. $1/(\ell d)$
\State $H' = +4W$ subsetwise $(R \times R)$-spanner~\cite{ahmed2020weighted}
\State add $H'$ to $H$
\EndIf
\EndFor
\EndWhile
\State add the $m'$ missing edges
\State \Return{$H$}
\end{algorithmic}
\end{algorithm}

\section{Multi-level spanners}\label{sec:multi_level}
Here we study a multi-level variant of graph spanners. We first define the problem:

\begin{definition}[Multi-level weighted additive spanner] \label{def:mlas}
Given a weighted graph $G(V,E)$ with maximum weight $W$, a nested sequence of subsets of vertices $S_{\ell} \subseteq S_{\ell-1} \subseteq \ldots \subseteq S_1 \subseteq V$, and $c \ge 0$, a multi-level additive spanner is a nested sequence of subgraphs $G_{\ell} \subseteq G_{\ell-1} \subseteq \ldots \subseteq G_1 \subseteq G$, where $G_i$ is a subsetwise $+cW$ spanner over $S_i$.
\end{definition}

Observe that Definition~\ref{def:mlas} generalizes the subsetwise spanner, which is a special case where $\ell=1$. We measure the size of a multi-level spanner by its \emph{sparsity}, defined by $\text{sparsity}(\{G_i\}_{i=1}^{\ell}) := \sum_{i=1}^{\ell} |E(G_i)|.$

The problem can equivalently be phrased in terms of priorities and rates: each vertex $v \in S_1$ has a priority $P(v)$ between 1 and $\ell$ (namely, $P(v) = \max \{i: v \in S_i\}$), and we wish to compute a single subgraph containing edges of different rates such that for all $u, v \in S_1$, there is a $+cW$ spanner path consisting of edges of rate at least $\min\{P(u), P(v)\}$. With this, we will typically refer to the \emph{priority} of $v$ to denote the highest $i$ such that $v \in S_i$, or 0 if $v \not\in S_1$.
In this section, we show that the multi-level version is not significantly harder than the ordinary ``single-level'' version: a subroutine which can compute an additive spanner can be used to compute a multi-level spanner whose sparsity is comparably good. Let $\OPT$ denote the minimum sparsity over all candidate multi-level additive spanners.

We first describe a simple rounding-up approach based on an algorithm by Charikar et al.~\cite{Charikar2004ToN} for the QoSMT problem, a similar generalization of the Steiner tree problem. For this approach, assume we have a subroutine which computes an exact or approximate single-level subsetwise spanner. Given $v \in S_1$, let $P(v) \in [1, \ell]$ denote the priority of $v$. The rounding-up approach is as follows: for each $v$, round $P(v)$ up to the nearest power of 2. This effectively constructs a ``rounded-up'' instance where all vertices in $S_1$ have priority either 1, 2, 4, \ldots, $2^{\lceil \log_2 \ell \rceil}$. The sparsity of the optimum solution in the rounded-up instance is at most $2\OPT$; given the optimum solution to the original instance with sparsity $\OPT$, a feasible solution to the rounded-up instance with sparsity at most $2\OPT$ can be obtained by rounding up the rate of each edge to the nearest power of 2.

For each $i \in \{1,2,4,\ldots,2^{\lceil \log_2 \ell \rceil}\}$, use the subroutine to compute a level-$i$ subsetwise spanner over all vertices whose rounded-up priority is at least $i$. The final multi-level additive spanner is obtained by taking the union of these computed spanners, by keeping an edge at the highest level it appears in. 

\begin{theorem}
\label{thrm:round_up_opt}
Assuming an exact subsetwise spanner subroutine, the solution computed by the rounding-up approach has sparsity at most $4 \cdot \OPT$.
\end{theorem}

\begin{figure}
\captionsetup[subfigure]{justification=centering}
\minipage{0.48\textwidth}
  \includegraphics[width=\linewidth]{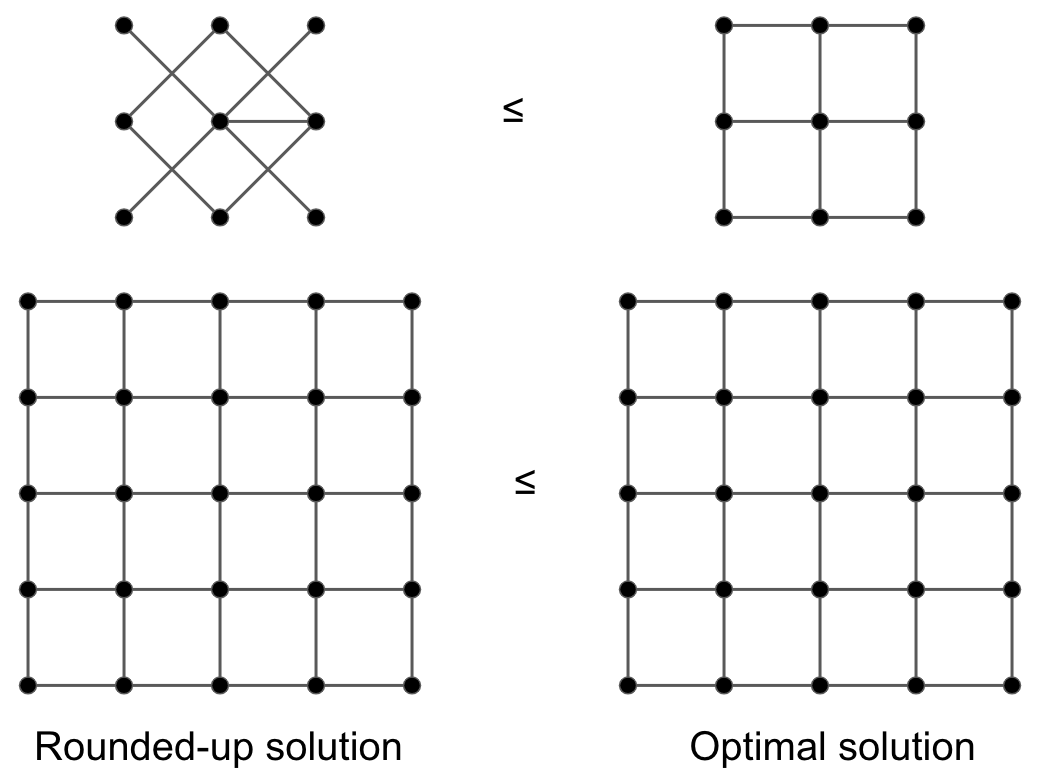}
  \subcaption{}
\label{figure:independent_computations}
\endminipage\hfill
\minipage{0.48\textwidth}
  \includegraphics[width=\linewidth]{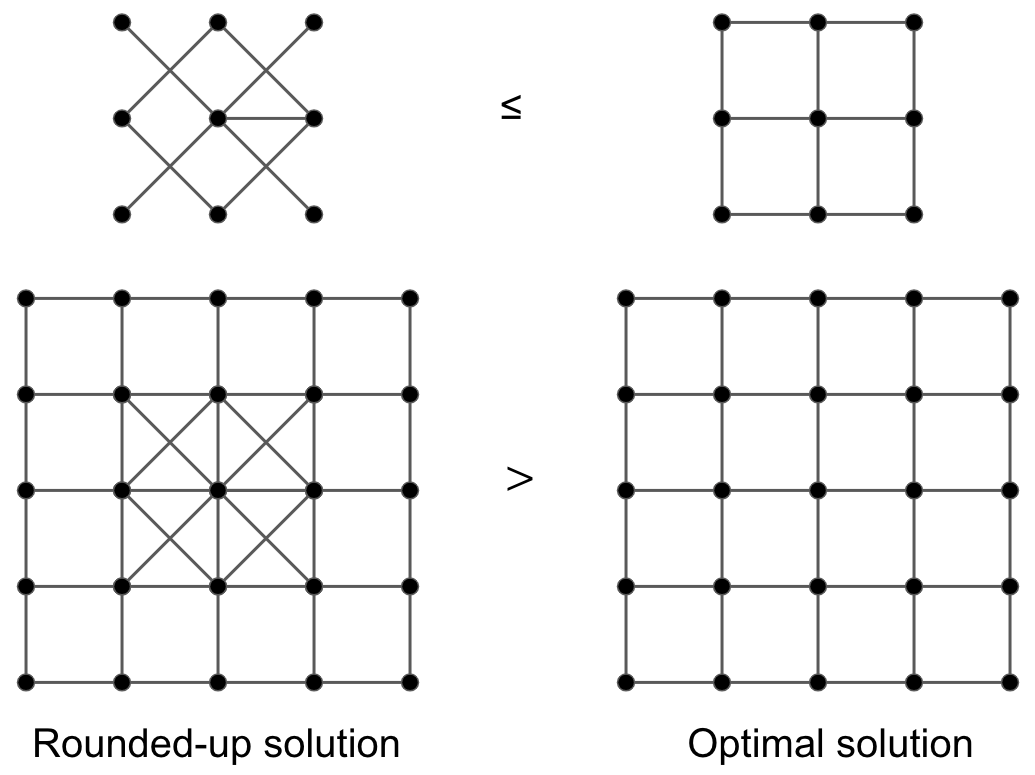}
  \subcaption{}
\label{figure:reuse_edge}
\endminipage
\caption{(a) The rounding-up approach computes an optimal spanner at each level (assuming an exact subroutine), so the sizes of the spanners on each level are at most that of the optimal solution ($9+40$ edges vs. $12+40$). (b) However, when an edge is present in a top-level solution, it must be present in lower-level solutions. The rounding-up approach takes the union of the spanners in the bottom level; in this case, the sparsity of the rounded-up solution ($9+48$ vs. $12+40$) is greater than that of the optimum.}
\end{figure}

This is proved using the same ideas as the $4\rho$-approximation for QoSMT~\cite{Charikar2004ToN}. As mentioned earlier, in practice we use an approximation algorithm to compute the subsetwise spanner instead of computing the minimum spanner.

\begin{theorem}
\label{thrm:round_up_apprx}
There exists a $\tilde{O}(n/\sqrt{|S_1|})$-approximation algorithm to compute multi-level weighted additive spanners with additive stretch $2W$ when $W = O(\log n)$.
\end{theorem}
This follows from using the $+2W$ subsetwise construction in Section 2. The approximation ratio of this subsetwise spanner algorithm is $O(nW/\sqrt{|S|})$ as the construction produces a spanner of size $O(nW\sqrt{|S|})$, while the sparsest additive spanner trivially has at least $|S|-1 = \Omega(|S|)$ edges.

We now show that, under certain conditions, if we have a subroutine which computes a subsetwise spanner of $G$, $S$ of size $O(n^a |S|^b)$ where $a$ and $b$ are absolute constants, a very na\"{i}ve algorithm can be used to obtain a multi-level spanner also with sparsity $O(n^a |S_1|^b)$.

\begin{theorem}
Suppose there is an absolute constant $0 < \alpha < 1$ such that $|S_i| \le \alpha |S_{i-1}|$ for all $i \in \{1,\ldots,\ell\}$. Then we can compute a multi-level spanner with sparsity $O(n^a |S_1|^b)$.
\end{theorem}
\begin{proof}
Consider the following simple construction: for each $i \in \{1,2,3,\ldots,\ell\}$, compute a level-$i$ subsetwise spanner of size $O(n^a |S_i|^b)$. Consider the union of these spanners, by keeping each edge at the highest level it appears. The sparsity of the returned multi-level spanner is at most
\begin{align*}
\text{sparsity}(\{G_i\}) &= O(n^a |S_1|^b + 2n^a |S_2|^b + 3n^a |S_3|^b + \ldots + \ell n^a |S_{\ell}|^b) \\
&\le O(n^a |S_1|^b (1 + 2\alpha^b + 3\alpha^{2b} + \ldots + \ell \alpha^{(\ell-1)b})) \\
&= O(n^a |S_1|^b)
\end{align*}
where we used the arithmetico-geometric series $1 + 2(\alpha^b) + 3(\alpha^b)^2 + \ldots = \frac{1}{(1 - \alpha^b)^2}$ which is constant for fixed $\alpha$, $b$. Note that $0 < \alpha < 1$ and $b > 0$, which implies $0 < \alpha^b < 1$.
\end{proof}
The assumption that $|S_i| \le \alpha|S_{i-1}|$ for some constant $\alpha$ is fairly natural, as many realistic networks tend to have significantly fewer hubs than non-hubs.

\begin{corollary}
Under the assumption $|S_i| \le \alpha |S_{i-1}|$ for all $i \in \{2,\ldots,\ell\}$, there exists a poly-time algorithm which computes a multi-level $+2$ spanner of sparsity $O(n\sqrt{|S_1|})$.
\end{corollary}
\begin{proof}
This follows by using the $+2$ construction by Cygan et al.~\cite{Cygan13} on $O(n\sqrt{|S|})$ edges as the subroutine.
\end{proof}

\section{Exact algorithm}\label{sec:exact_algo}

To compute a minimum size additive spanner, we utilize a slight modification of the ILP in \cite[Section 9]{ahmed2020graphElsevier}, wherein we choose the specific distortion function $f(t) = t+cW$ and minimize the sparsity rather than total weight of the spanner.  For completeness, we present the full ILP for computing a single-level additive subsetwise spanner below along with a brief description of the multi-level extension. Here $E'$ represents the bidirected edge set, obtained by adding directed edges $(u,v)$ and $(v,u)$ for each edge $uv \in E$. The binary variable $x_{(i,j)}^{uv}$ is 1 if edge $(i,j)$ is included on the selected $u$-$v$ path and 0 otherwise, and $w(e)$ is the weight of edge $e$.


\begin{align}
    \allowdisplaybreaks
    \text{Minimize} \sum_{e \in E} x_e \text{ subject to}\\
    \sum_{(i,j) \in E'} x_{(i,j)}^{uv} w(e) &\le \dist_G(u,v) + c W & \hspace{-10pt}\forall (u,v) \in S; e = ij \label{ineq:ilp-spanner}\\
    \sum_{(i,j) \in Out(i)} x_{(i,j)}^{uv} - \sum_{(j,i) \in In(i)} x_{(j,i)}^{uv} &= \begin{cases}
        1 & i = u \\
        -1 & i = v \\
        0 & \text{else}
    \end{cases} & \hspace{-10pt} \forall (u,v) \in S; \forall i \in V \label{ineq:ilp-flow} \\
    \sum_{(i,j) \in Out(i)} x_{(i,j)}^{uv} & \le 1 &\forall (u,v) \in S; \forall i \in V \label{ineq:ilp-out}\\
    x_{(i,j)}^{uv} + x_{(j,i)}^{uv} &\le x_e & \hspace{-30pt}\forall (u,v) \in S; \forall e = \{i,j\} \in E \label{ineq:ilp-edge}\\
    x_e, x_{(i,j)}^{uv} &\in \{0,1\}
\end{align}






Inequalities~\eqref{ineq:ilp-flow}--\eqref{ineq:ilp-out} enforce that for each $u$, $v \in S$, the selected edges corresponding to $u$, $v$ form a path; inequality~\eqref{ineq:ilp-spanner} enforces that the length of this path is at most $\dist_G(u,v) + cW$ (note that $W$ may be replaced with $W(u,v)$). Inequality~\eqref{ineq:ilp-edge} ensures that if $x_{(i,j)}^{uv} = 1$ or $x_{(i,j)}^{uv} = 1$, then edge $ij$ is taken.

To generalize the ILP formulation to the multi-level problem, we take a similar set of variables for every level.
The rest of the constraints are similar, except we define $x_e^{k} = 1$ if edge $e$ is present on level $k$ and the variables $x_{(i,j)}^{uv}$ are also indexed by level. We add the constraint $x_e^k \le x_e^{k-1}$ for all $k \in \{2,\ldots,\ell\}$ which enforces that if edge $e$ is present on level $k$, it is also present on all lower levels. Finally, the objective is to minimize the sparsity $\sum_{k=1}^{\ell} \sum_{e \in E}x_e^k$.

\section{Experiments}

In this section, we provide experimental results involving the rounding-up framework described in Section~\ref{sec:multi_level}. This framework needs a single level subroutine; we use the $+2W$ subsetwise construction in Section~\ref{sec:subsetwise} and the three pairwise $+2W(\cdot, \cdot)$, $+4W(\cdot, \cdot)$, $+6W$ constructions provided in~\cite{ahmed2020weighted}\footnote{Note that, one can show that the $+2W$, $+4W$, $+8W$ spanners in~\cite{ahmed2020weighted} are actually $+2W(.,.), +4W(.,.)$ and $+6W$ spanners respectively by using a tighter analysis~\cite{ahmed2021weighted}.} (see Appendix~\ref{section:pairwise}).
We generate multi-level instances and solve the instances using our exact algorithm and the four approximation algorithms. We consider natural questions about how the number of levels $\ell$, number of vertices $n$, and decay rate of terminals with respect to levels affect the running times and (experimental) approximation ratios, defined as the sparsity of the returned multi-level spanner divided by $\OPT$.

We used CPLEX 12.6.2 as an ILP solver in a high-performance computer for all experiments (Lenovo NeXtScale nx360 M5 system with 400 nodes). 
Each node has 192 GB of memory. We have used Python for implementing the algorithms and spanner constructions. Since we have run the experiment on a couple of thousand instances, we run the solver for four hours.

\subsection{Experiment Parameters}

We run experiments first to test experimental approximation ratio vs. the parameters, and then to test running time vs. parameters.  Each set of experiments has several parameters: the graph generator, the number of levels $\ell$, the number of vertices $n$, and how the size of the terminal sets $S_i$ (vertices requiring level or priority at least $i$) decrease as $i$ decreases.

In what follows, we use the Erd\H{o}s--R\'{e}nyi (ER)~\cite{erdos1959random}, Watts--Strogatz (WS)~\cite{watts1998collective}, Barab\'{a}si--Albert (BA)~\cite{barabasi1999emergence}, and random geometric (GE)~\cite{penrose2003random} models. 
Let $p$ be the edge selection probability. If we set $p=(1+\varepsilon)\frac{\ln n}{n}$, then the generated Erd\H{o}s--R\'{e}nyi graph is connected with high probability for $\varepsilon>0$~\cite{erdos1959random}). For our experiments we use $\varepsilon = 1$.
In the Watts-Strogatz model, we initially create a ring lattice of constant degree $K$. For our experiments we use $K=6$ and $p=0.2$.
In the Barab\'{a}si--Albert model, a new node is connected to $m$ existing nodes. For our experiments we use $m=5$.
In the random geometric graph model, two nodes are connected to each other if their Euclidean distance is not larger than a threshold~$r_c$. For $r_c=\sqrt{\frac{(1+\epsilon)\ln n}{\pi n}}$ with $\epsilon>0$, the synthesized graph is connected with a high probability\cite{penrose2003random}.
We generate a set of small graphs ($10 \le n \le 40$) and a set of large graphs ($50 \le n \le 500$). We only compute the exact solutions for the small graphs since the ILP has an exponential running time.
In this paper, we provide the results of Erd\H{o}s--R\'{e}nyi graphs since it is the most popular model. However, the radius\footnote{The minimum over all $v \in V$ of $\max_{w \in V} d_G(v,w)$ where $d_G(v,w)$ is the graph distance (by number of edges, not total weight) between $v$ and $w$} of Erd\H{o}s--R\'{e}nyi graphs is relatively small. In our dataset, the range of the radius is 2-4. Hence, we also provide the results of random geometric graphs which have larger radius (4-12). The remaining results and the radius distribution of different generators are available at the supplement Github link.
We consider number of levels $\ell\in\{1,2,3\}$ for small graphs, $\ell\in\{1,\dots,10\}$ for large graphs, and adopt two methods for selecting terminal sets: \emph{linear} and \emph{exponential}.
A terminal set $S_1$ with lowest priority of size $n(1-\frac{1}{\ell+1})$ in the linear case and $\frac{n}{2}$ in the exponential case is chosen uniformly at random.  For each subsequent level, $\frac{1}{\ell+1}$ vertices are deleted at random in the linear case, whereas half the remaining vertices are deleted in the exponential case.  Levels/priorities and terminal sets are related via $S_i = \{v \in S_1: P(v) \ge i\}$. We choose edge weights $w(e)$ randomly, independently, and uniformly from $\{1,2,3\dots,10\}$.

An experimental instance of the multi-level problem here is thus characterized by four parameters: graph generator, number of vertices $n$, number of levels $\ell$, and terminal selection method $\textsc{TSM}\in\{\textsc{Linear,Exponential}\}$. As there is randomness involved, we generated five instances for every choice of parameters (e.g., ER, $n = 30$, $\ell=2$, \textsc{Linear}).
For each instance of the small graphs, we compute the approximate solution using either the $+2W$, $+2W(\cdot, \cdot)$, $+4W$, or $+6W$ spanner subroutine, and the exact solution using the ILP described in Section~\ref{sec:exact_algo}. We compute the experimental approximation ratio by dividing the sparsity of the approximate solution by the sparsity of the optimum solution ($\OPT$). For large graphs, we only compute the approximate solution.


\subsection{Results}

We consider different spanner constructions as the single level subroutine in the multi-level spanner. We first consider the $+2W$ subsetwise construction (Section~\ref{sec:subsetwise}).

\subsubsection{The $+2W$ Subsetwise Construction-based Approximation}

We first describe the experimental results on Erd\H{o}s--R{\'e}nyi graphs w.r.t. $n$, $\ell$, and terminal selection method in Figure~\ref{LinePlots_ER_SUB}. The average experimental ratio increases as $n$ increases. This is expected since the theoretical approximation ratio of $\tilde{O}(n/\sqrt{|S_1|})$ is proportional to $n$. The average and minimum experimental ratio does not change that much as the number of levels increases; however, the maximum ratio increases. The experimental ratio of the linear terminal selection method is slightly better compared to that of the exponential method.

\begin{figure}[ht]
\begin{minipage}{\textwidth}
    \centering
    \begin{subfigure}[b]{0.30\textwidth}
        \includegraphics[width=\textwidth]{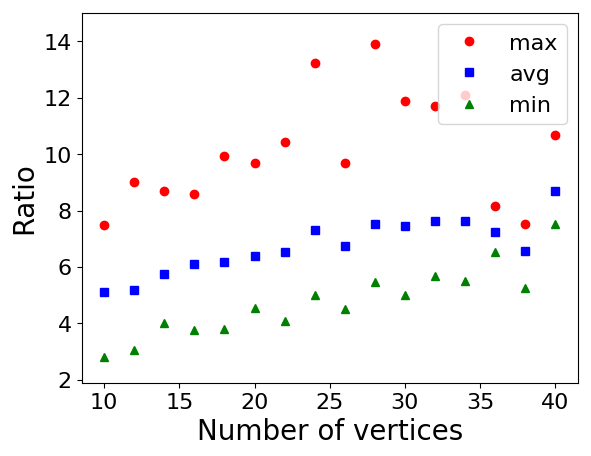}
    \end{subfigure}
    ~ 
    \begin{subfigure}[b]{0.30\textwidth}
        \includegraphics[width=\textwidth]{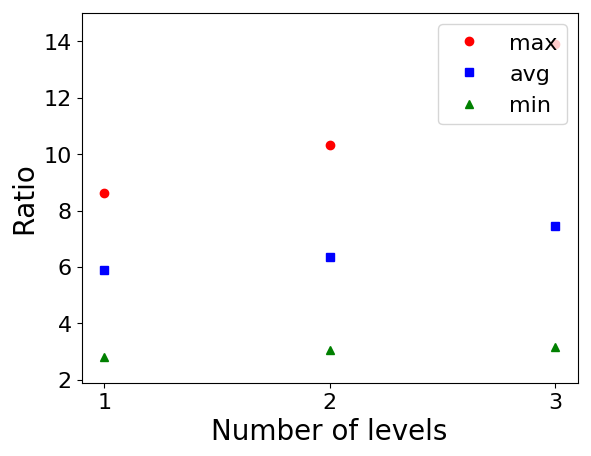}
    \end{subfigure}
    ~
    \begin{subfigure}[b]{0.30\textwidth}
        \includegraphics[width=\textwidth]{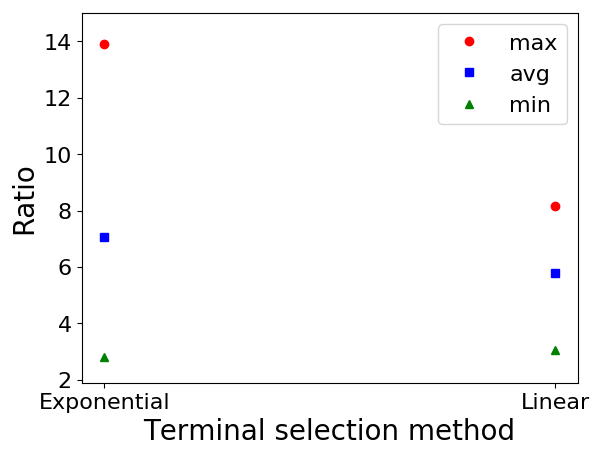}
    \end{subfigure}
    \caption{Performance of the algorithm that uses $+2W$ subsetwise spanner as the single level solver on Erd\H{o}s--R{\'e}nyi graphs w.r.t.\
      $n$, $\ell$, and terminal selection method.}
    \label{LinePlots_ER_SUB}
\end{minipage}
\end{figure}

We describe the experimental results on random geometric graphs w.r.t. $n$, $\ell$, and terminal selection methods in Figure~\ref{LinePlots_GE_SUB}.
In both cases the average ratio increases as $n$ and $\ell$ increases. The average ratio is relatively lower for the linear terminal selection method.

\begin{figure}[H]

\begin{minipage}{\textwidth}
    \centering
    \begin{subfigure}[b]{0.30\textwidth}
        \includegraphics[width=\textwidth]{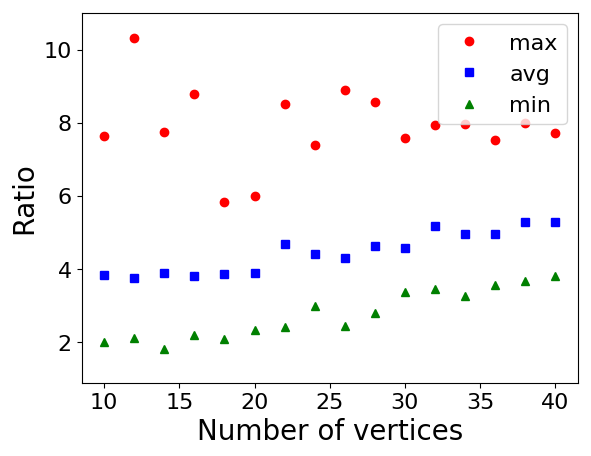}
    \end{subfigure}
    ~ 
    \begin{subfigure}[b]{0.30\textwidth}
        \includegraphics[width=\textwidth]{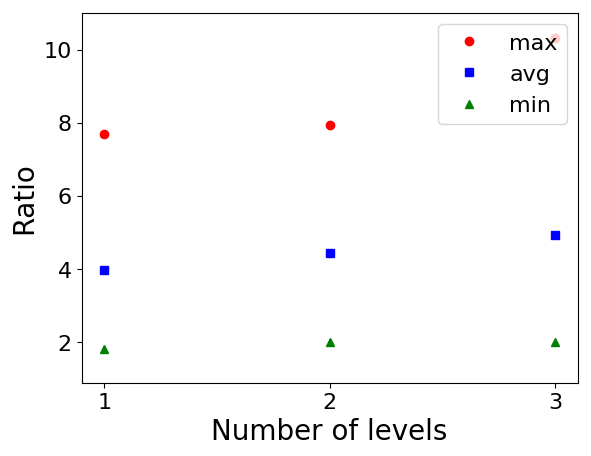}
    \end{subfigure}
    ~
    \begin{subfigure}[b]{0.30\textwidth}
        \includegraphics[width=\textwidth]{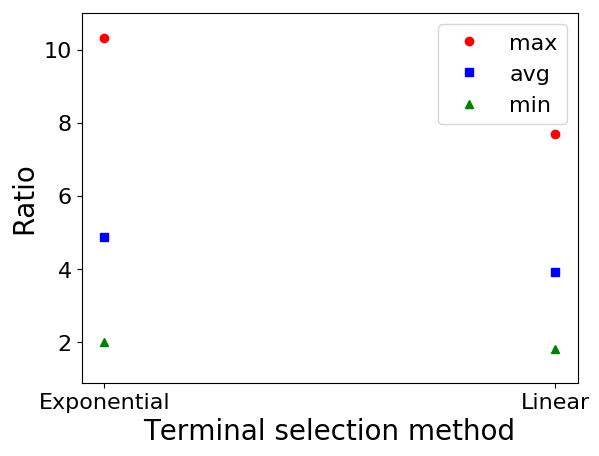}
    \end{subfigure}
    \caption{Performance of the algorithm that uses $+2W$ subsetwise spanner as the single level solver on random geometric graphs w.r.t.\
      $n$, $\ell$, and terminal selection method.}
    \label{LinePlots_GE_SUB}
\end{minipage}
\end{figure}

The plots of  Watts--Strogatz and Barab\'{a}si--Albert graphs are available in the Github repository.

\subsubsection{The $+2W(\cdot, \cdot)$ Pairwise Construction-based Approximation}

We now consider the $+2W(\cdot, \cdot)$ pairwise construction~\cite{ahmed2020weighted} (Algorithm~\ref{alg:2W-pairwise}). We first describe the experimental results on Erd\H{o}s--R{\'e}nyi graphs w.r.t. $n$, $\ell$, and terminal selection method in Figure~\ref{LinePlots_ER_PAIR_2W}. The average experimental ratio increases as $n$ increases. This is expected since the theoretical approximation ratio is proportional to $n$. The average and minimum experimental ratio do not change that much as the number of levels increases, however, the maximum ratio increases. The experimental ratio of the linear terminal selection method is also slightly better compared to that of the exponential method.

\begin{figure}[ht]
\begin{minipage}{\textwidth}
    \centering
    \begin{subfigure}[b]{0.30\textwidth}
        \includegraphics[width=\textwidth]{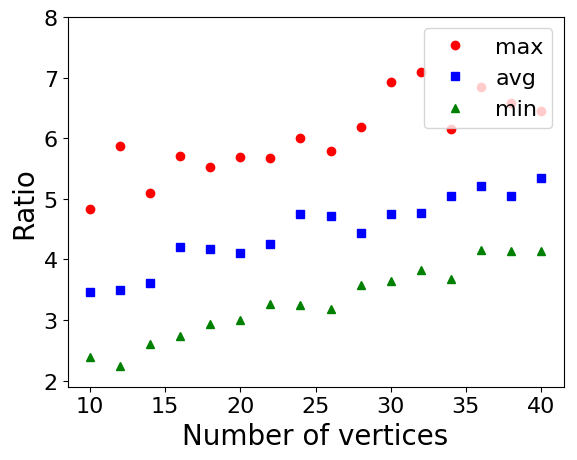}
    \end{subfigure}
    ~ 
    \begin{subfigure}[b]{0.30\textwidth}
        \includegraphics[width=\textwidth]{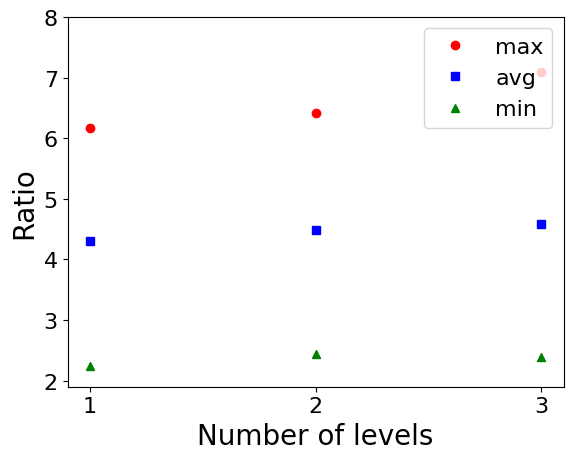}
    \end{subfigure}
    ~
    \begin{subfigure}[b]{0.30\textwidth}
        \includegraphics[width=\textwidth]{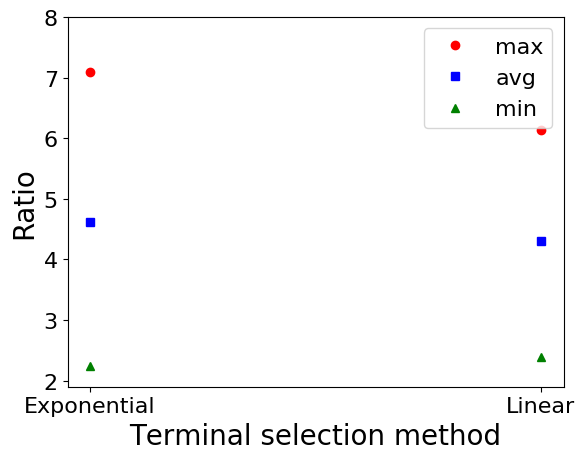}
    \end{subfigure}
    \caption{Performance of the algorithm that uses $+2W(\cdot, \cdot)$ pairwise spanner as the single level solver on Erd\H{o}s--R{\'e}nyi graphs w.r.t.\
      $n$, $\ell$, and terminal selection method.}
    \label{LinePlots_ER_PAIR_2W}
\end{minipage}
\end{figure}

We describe the experimental results on random geometric graphs w.r.t. $n$, $\ell$, and terminal selection method in  Figure~\ref{LinePlots_GE_PAIR_2W}. 
The average experimental ratio increases as $n$ increases. The maximum ratio increases as $\ell$ increases. Again, the experimental ratio of the linear terminal selection method is relatively smaller compared to the exponential method.

\begin{figure}[H]

\begin{minipage}{\textwidth}
    \centering
    \begin{subfigure}[b]{0.30\textwidth}
        \includegraphics[width=\textwidth]{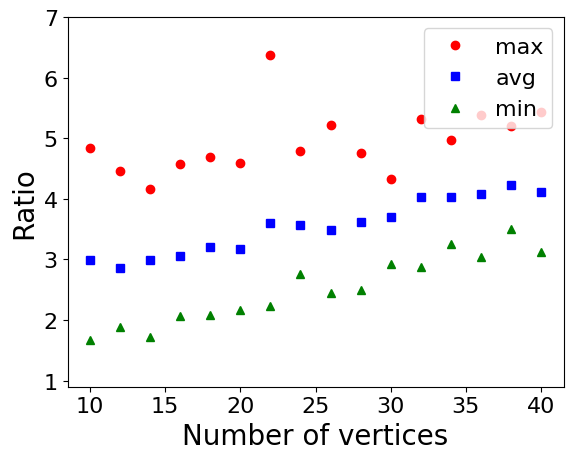}
    \end{subfigure}
    ~ 
    \begin{subfigure}[b]{0.30\textwidth}
        \includegraphics[width=\textwidth]{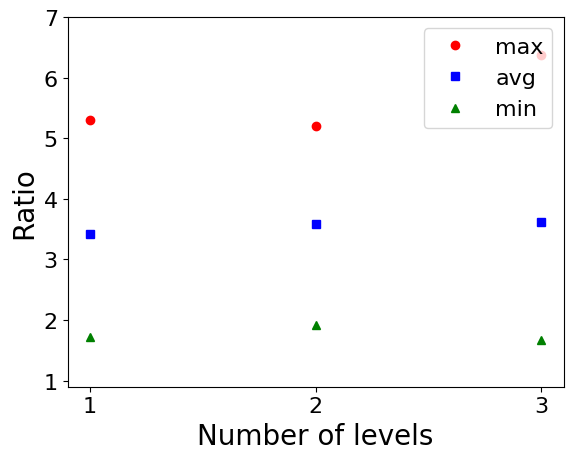}
    \end{subfigure}
    ~
    \begin{subfigure}[b]{0.30\textwidth}
        \includegraphics[width=\textwidth]{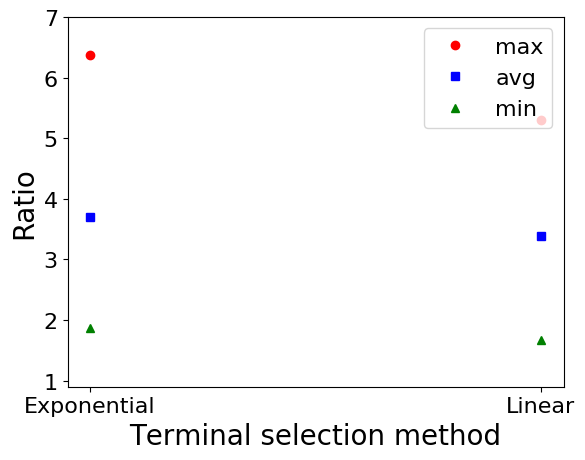}
    \end{subfigure}
    \caption{Performance of the algorithm that uses $+2W(\cdot, \cdot)$ pairwise spanner as the single level solver on random geometric graphs w.r.t.\
      $n$, $\ell$, and terminal selection method.}
    \label{LinePlots_GE_PAIR_2W}
\end{minipage}
\end{figure}

\subsubsection{Comparison between Global and Local Setups}
One major difference between the subsetwise and pairwise construction is the subsetwise construction considers the (global) maximum edge weight $W$ of the graph in the error. On the other hand, the $+cW(\cdot, \cdot)$ spanners consider the (local) maximum edge weight in a shortest path for each pair of vertices $s, t$. We provide a comparison between the global and local settings.

We describe the experimental results on Erd\H{o}s--R{\'e}nyi graphs w.r.t. $n$, $\ell$, and the terminal selection method in Figure~\ref{BoxPlots_ER_GLOB_LOC}. The average experimental ratio increases as $n$ increases for both global and local settings. However, the ratio of the local setting is smaller compared to that of the global setting. One reason for this difference is the solution to the global exact algorithm is relatively smaller since the global setting considers larger errors. The ratio of the global setting increases as the number of levels increases and for the exponential terminal selection method. For the local setting, the ratio does not change that much. 

\begin{figure}[ht]
\begin{minipage}{\textwidth}
    \centering
    \begin{subfigure}[b]{0.30\textwidth}
        \includegraphics[width=\textwidth]{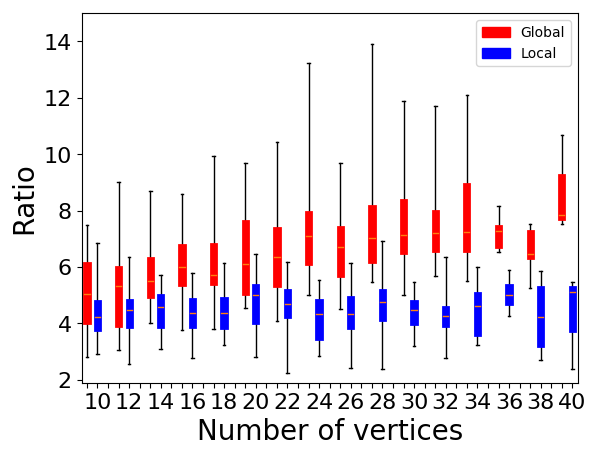}
    \end{subfigure}
    ~ 
    \begin{subfigure}[b]{0.30\textwidth}
        \includegraphics[width=\textwidth]{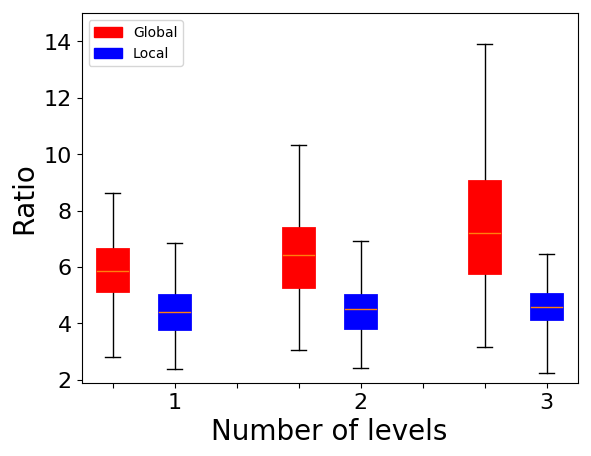}
    \end{subfigure}
    ~
    \begin{subfigure}[b]{0.30\textwidth}
        \includegraphics[width=\textwidth]{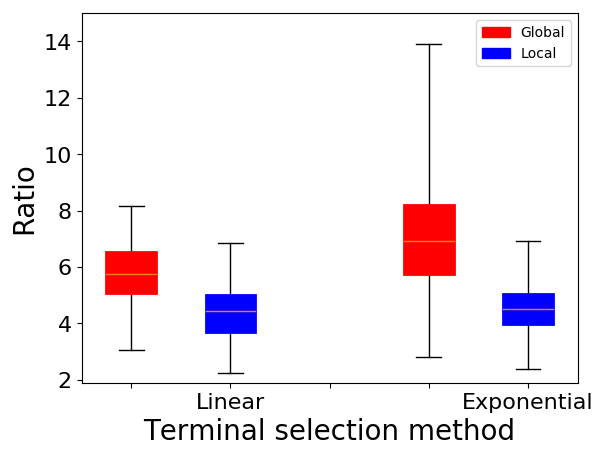}
    \end{subfigure}
    \caption{Performance of the global and local construction-based algorithms on Erd\H{o}s--R{\'e}nyi graphs w.r.t.\
      $n$, $\ell$, and terminal selection method.}
    \label{BoxPlots_ER_GLOB_LOC}
\end{minipage}
\end{figure}

We describe the experimental results on random geometric graphs w.r.t. $n$, $\ell$, and the terminal selection method in Figure~\ref{BoxPlots_GE_GLOB_LOC}. 
The ratio of the local setting is smaller compared to the global setting.

\begin{figure}[H]

\begin{minipage}{\textwidth}
    \centering
    \begin{subfigure}[b]{0.30\textwidth}
        \includegraphics[width=\textwidth]{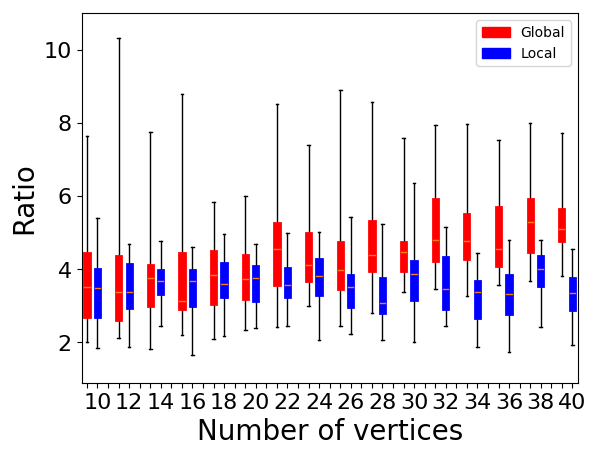}
    \end{subfigure}
    ~ 
    \begin{subfigure}[b]{0.30\textwidth}
        \includegraphics[width=\textwidth]{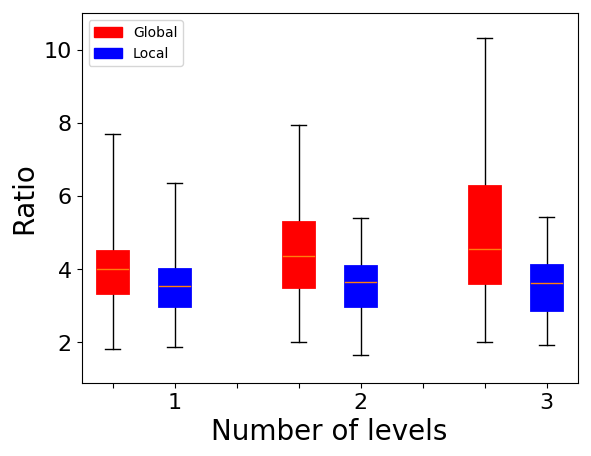}
    \end{subfigure}
    ~
    \begin{subfigure}[b]{0.30\textwidth}
        \includegraphics[width=\textwidth]{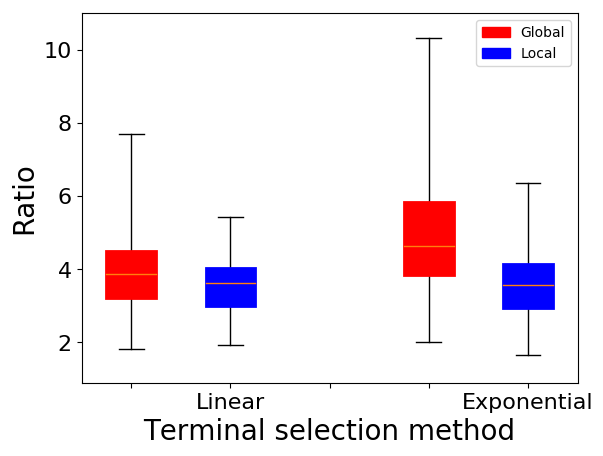}
    \end{subfigure}
    \caption{Performance of the global and local construction-based algorithms on random geometric graphs w.r.t.\
      $n$, $\ell$, and terminal selection method.}
    \label{BoxPlots_GE_GLOB_LOC}
\end{minipage}
\end{figure}

\subsubsection{The $+4W(\cdot, \cdot)$ Pairwise Construction-based Approximation}

We now consider the $+4W(\cdot, \cdot)$ pairwise construction~\cite{ahmed2020weighted} (Algorithm~\ref{alg:4W-pairwise}) as a single level subroutine. We first describe the experimental results on Erd\H{o}s--R{\'e}nyi graphs w.r.t. $n$, $\ell$, and terminal selection method in Figure~\ref{LinePlots_ER_PAIR_4W}. The average experimental ratio increases as $n$ increases. This is expected since the theoretical approximation ratio is proportional to $n$. The average experimental ratio does not change that much as the number of levels increases; however, the maximum ratio increases. The experimental ratio of the linear terminal selection method is also slightly better  compared to that of the exponential method.

\begin{figure}[ht]
\begin{minipage}{\textwidth}
    \centering
    \begin{subfigure}[b]{0.30\textwidth}
        \includegraphics[width=\textwidth]{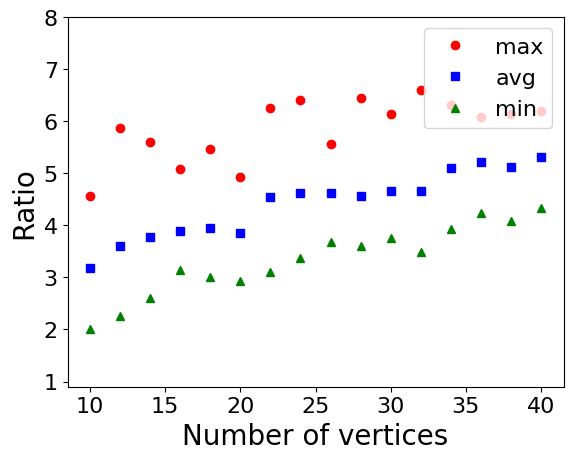}
    \end{subfigure}
    ~ 
    \begin{subfigure}[b]{0.30\textwidth}
        \includegraphics[width=\textwidth]{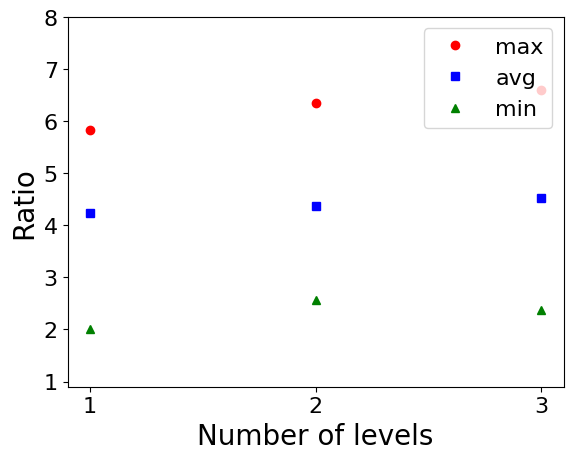}
    \end{subfigure}
    ~
    \begin{subfigure}[b]{0.30\textwidth}
        \includegraphics[width=\textwidth]{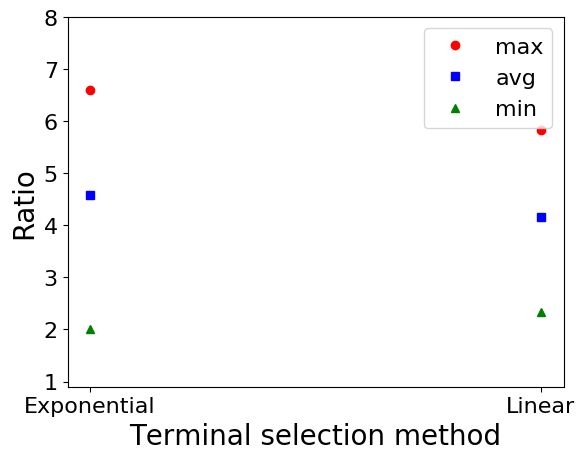}
    \end{subfigure}
    \caption{Performance of the algorithm that uses $+4W(\cdot, \cdot)$ pairwise spanner as the single level solver on Erd\H{o}s--R{\'e}nyi graphs w.r.t.\
      $n$, $\ell$, and terminal selection method.}
    \label{LinePlots_ER_PAIR_4W}
\end{minipage}
\end{figure}

We describe the experimental results on random geometric graphs w.r.t. $n$, $\ell$, and terminal selection method in  Figure~\ref{LinePlots_GE_PAIR_4W}. 
The experimental ratio increases as the number of vertices increases. The maximum ratio increases as the number of levels increases. Again, the experimental ratio of the linear terminal selection method is relatively smaller compared to the exponential method.

\begin{figure}[H]

\begin{minipage}{\textwidth}
    \centering
    \begin{subfigure}[b]{0.30\textwidth}
        \includegraphics[width=\textwidth]{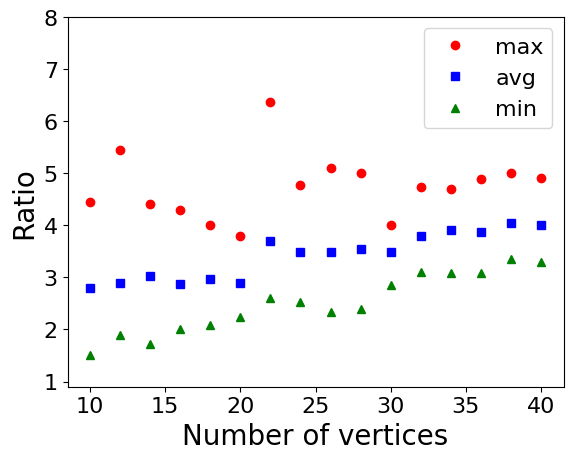}
    \end{subfigure}
    ~ 
    \begin{subfigure}[b]{0.30\textwidth}
        \includegraphics[width=\textwidth]{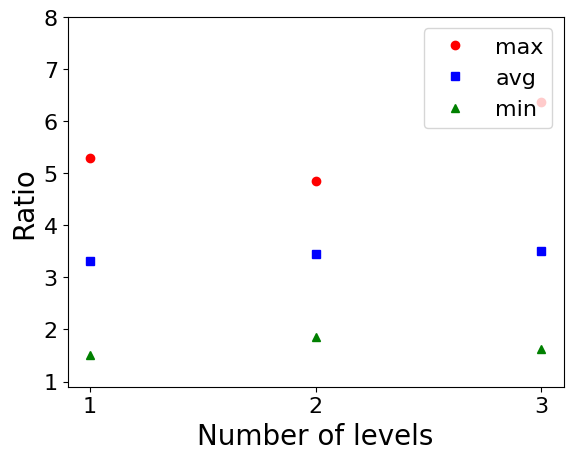}
    \end{subfigure}
    ~
    \begin{subfigure}[b]{0.30\textwidth}
        \includegraphics[width=\textwidth]{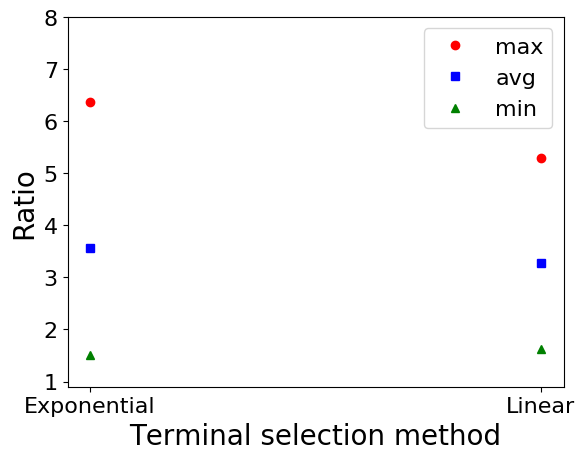}
    \end{subfigure}
    \caption{Performance of the algorithm that uses $+4W(\cdot, \cdot)$ pairwise spanner as the single level solver on random geometric graphs w.r.t.\
      $n$, $\ell$, and terminal selection method.}
    \label{LinePlots_GE_PAIR_4W}
\end{minipage}
\end{figure}

\subsubsection{Comparison between $+2W(\cdot, \cdot)$ and $+4W(\cdot, \cdot)$ Setups}
We now provide a comparison between the pairwise $+2W(\cdot, \cdot)$ and $+4W(\cdot, \cdot)$ construction-based approximation algorithms. We first describe the experimental results on Erd\H{o}s--R{\'e}nyi graphs w.r.t. $n$, $\ell$, and the terminal selection method in Figure~\ref{BoxPlots_ER_LOC_LOC}. The average experimental ratio increases as $n$ increases for both $+2W(\cdot, \cdot)$ and $+4W(\cdot, \cdot)$ settings. As $n$ increases, the ratio of the $+4W(\cdot, \cdot)$ construction-based algorithm decreases slightly. The $+4W(\cdot, \cdot)$ construction-based algorithm slightly outperforms the $+2W(\cdot, \cdot)$ algorithm for $\ell=3$ and exponential selection method.

\begin{figure}[ht]
\begin{minipage}{\textwidth}
    \centering
    \begin{subfigure}[b]{0.30\textwidth}
        \includegraphics[width=\textwidth]{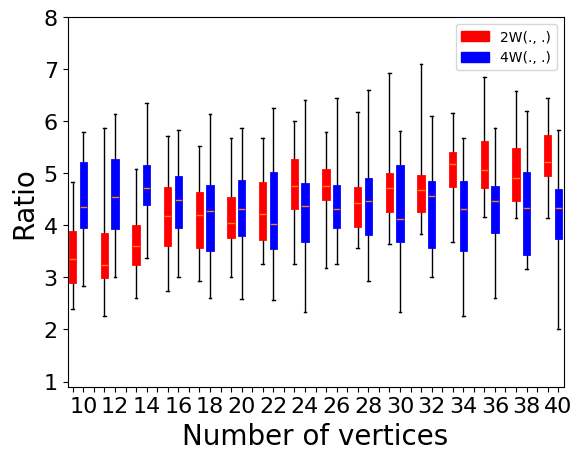}
    \end{subfigure}
    ~ 
    \begin{subfigure}[b]{0.30\textwidth}
        \includegraphics[width=\textwidth]{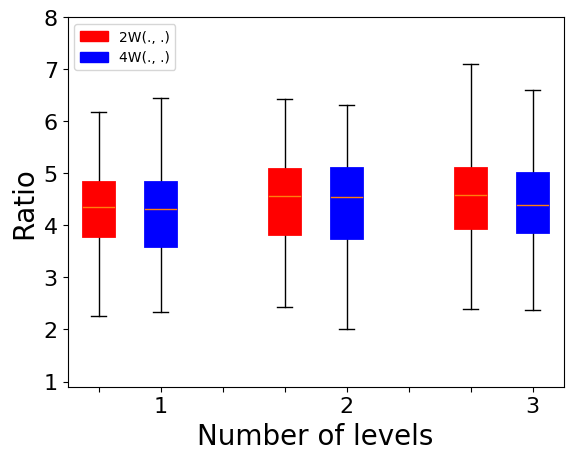}
    \end{subfigure}
    ~
    \begin{subfigure}[b]{0.30\textwidth}
        \includegraphics[width=\textwidth]{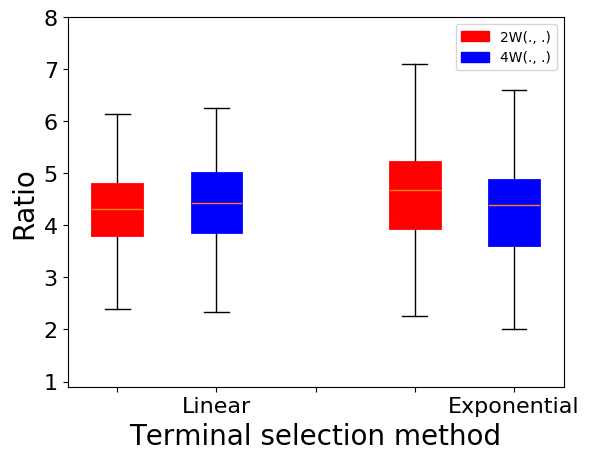}
    \end{subfigure}
    \caption{Performance of the pairwise $+2W(\cdot, \cdot)$ and $+4W(\cdot, \cdot)$ construction-based algorithms on Erd\H{o}s--R{\'e}nyi graphs w.r.t.\
      $n$, $\ell$, and terminal selection method.}
    \label{BoxPlots_ER_LOC_LOC}
\end{minipage}
\end{figure}

We describe the experimental results on random geometric graphs w.r.t. $n$, $\ell$, and the terminal selection method in Figure~\ref{BoxPlots_GE_LOC_LOC}. 
As $n$ increases the average ratio of $+4W(\cdot, \cdot)$-based approximation algorithm becomes smaller compared to the $+2W(\cdot, \cdot)$-based algorithm. The average ratio of $+4W(\cdot, \cdot)$ is relatively smaller for the exponential terminal selection method.

\begin{figure}[H]

\begin{minipage}{\textwidth}
    \centering
    \begin{subfigure}[b]{0.30\textwidth}
        \includegraphics[width=\textwidth]{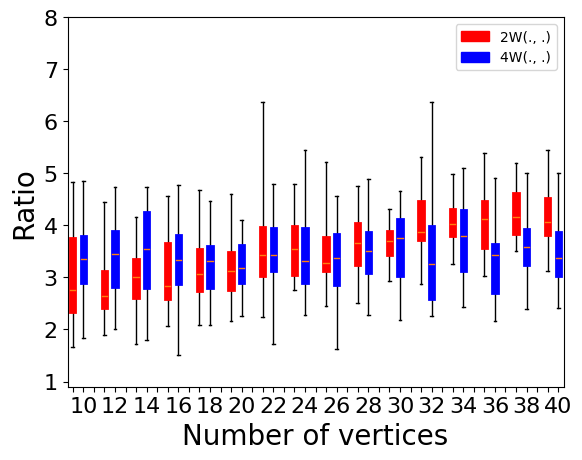}
    \end{subfigure}
    ~ 
    \begin{subfigure}[b]{0.30\textwidth}
        \includegraphics[width=\textwidth]{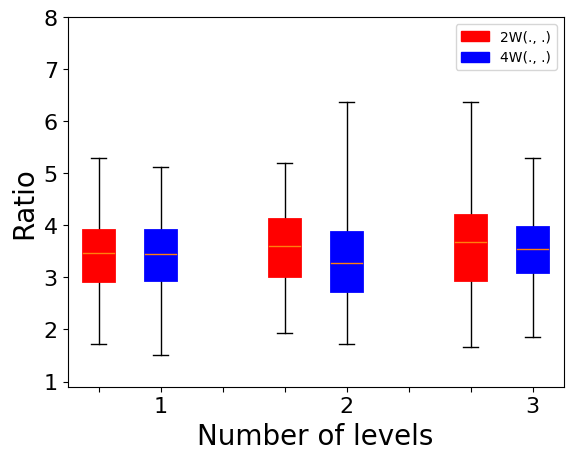}
    \end{subfigure}
    ~
    \begin{subfigure}[b]{0.30\textwidth}
        \includegraphics[width=\textwidth]{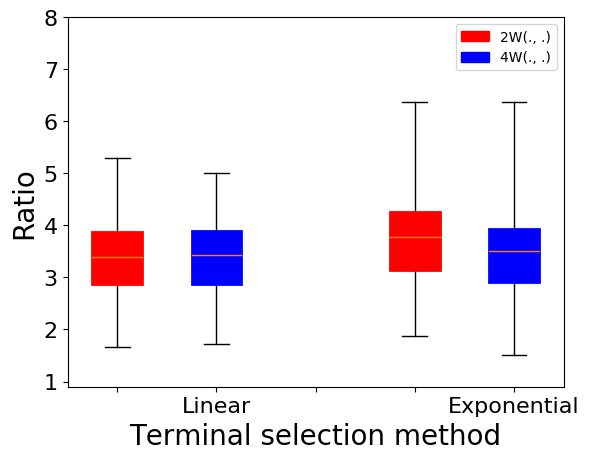}
    \end{subfigure}
    \caption{Performance of the pairwise $+2W(\cdot, \cdot)$ and $+4W(\cdot, \cdot)$ construction-based algorithms on random geometric graphs w.r.t.\
      $n$, $\ell$, and terminal selection method.}
    \label{BoxPlots_GE_LOC_LOC}
\end{minipage}
\end{figure}

\subsubsection{The $+6W$ Pairwise Construction-based Approximation}

We now consider the $+6W$ pairwise construction~\cite{ahmed2020weighted} (Algorithm~\ref{alg:8W-pairwise}) as a single level solver. We first describe the experimental results on Erd\H{o}s--R{\'e}nyi graphs w.r.t. $n$, $\ell$, and terminal selection method in Figure~\ref{LinePlots_ER_PAIR_6W}. The average experimental ratio increases as $n$ increases. This is expected since the theoretical approximation ratio is proportional to $n$. The average experimental ratio does not change that much as the number of levels increases; however, the maximum ratio increases. The maximum and average experimental ratios of the linear terminal selection method are slightly better compared to that of the exponential method.

\begin{figure}[ht]
\begin{minipage}{\textwidth}
    \centering
    \begin{subfigure}[b]{0.30\textwidth}
        \includegraphics[width=\textwidth]{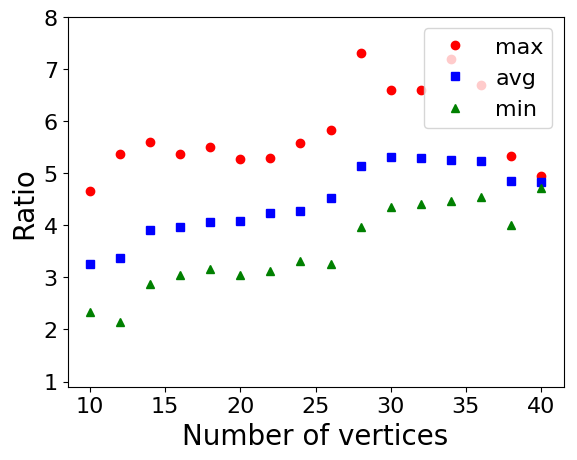}
    \end{subfigure}
    ~ 
    \begin{subfigure}[b]{0.30\textwidth}
        \includegraphics[width=\textwidth]{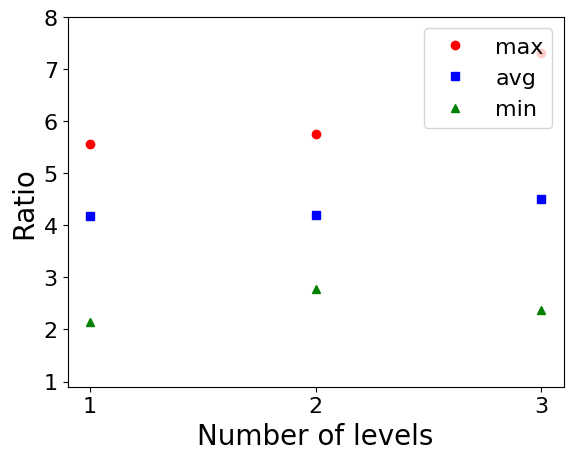}
    \end{subfigure}
    ~
    \begin{subfigure}[b]{0.30\textwidth}
        \includegraphics[width=\textwidth]{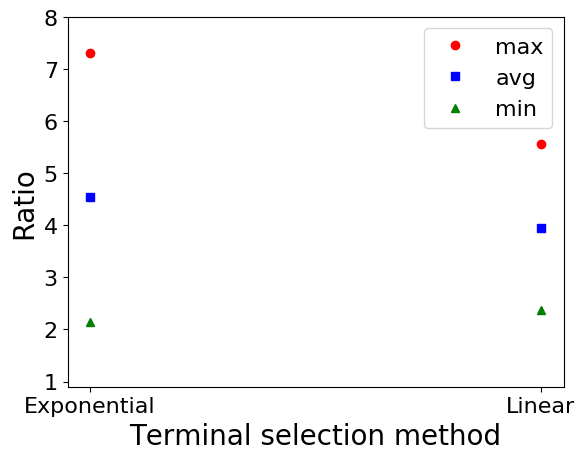}
    \end{subfigure}
    \caption{Performance of the algorithm that uses $+6W$ pairwise spanner as the single level solver on Erd\H{o}s--R{\'e}nyi graphs w.r.t.\
      $n$, $\ell$, and terminal selection method.}
    \label{LinePlots_ER_PAIR_6W}
\end{minipage}
\end{figure}

We describe the experimental results on random geometric graphs w.r.t. $n$, $\ell$, and terminal selection method in Figure~\ref{LinePlots_GE_PAIR_6W}. 
The experimental ratio increases as the number of vertices increases. The maximum ratio increases as the number of levels increases. Again, the experimental ratio of the linear terminal selection method is relatively smaller compared to the exponential method.

\begin{figure}[H]

\begin{minipage}{\textwidth}
    \centering
    \begin{subfigure}[b]{0.30\textwidth}
        \includegraphics[width=\textwidth]{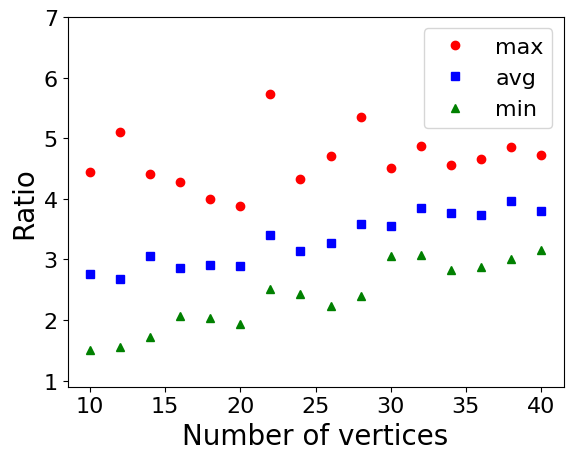}
    \end{subfigure}
    ~ 
    \begin{subfigure}[b]{0.30\textwidth}
        \includegraphics[width=\textwidth]{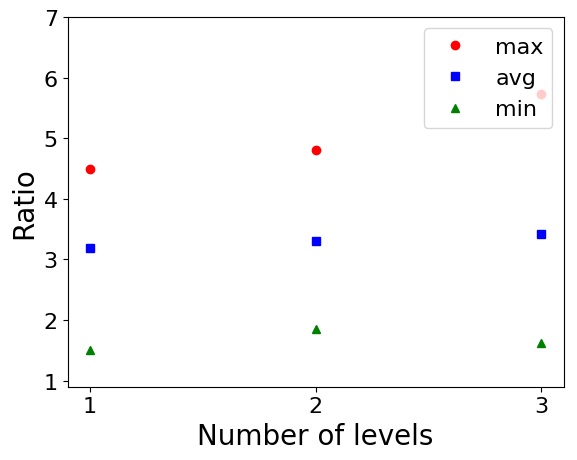}
    \end{subfigure}
    ~
    \begin{subfigure}[b]{0.30\textwidth}
        \includegraphics[width=\textwidth]{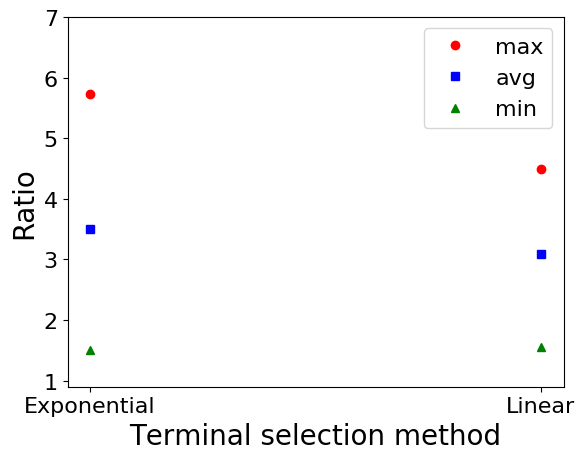}
    \end{subfigure}
    \caption{Performance of the algorithm that uses $+6W(\cdot, \cdot)$ pairwise spanner as the single level solver on random geometric graphs w.r.t.\
      $n$, $\ell$, and terminal selection method.}
    \label{LinePlots_GE_PAIR_6W}
\end{minipage}
\end{figure}

\subsubsection{Comparison between $+2W$ and $+6W$ Setups}
We now provide a comparison between pairwise $+2W$ and $+6W$ construction-based approximation algorithms. We first describe the experimental results on Erd\H{o}s--R{\'e}nyi graphs w.r.t. $n$, $\ell$, and the terminal selection method in Figure~\ref{BoxPlots_ER_GLOB_GLOB}. The average experimental ratio increases as $n$ increases for both $+2W$ and $+6W$ settings. As the number of vertices increases, the ratio of the $+6W$ construction-based algorithm gets smaller. This is expected since a larger error makes the problem easier to solve. Similarly, as $\ell$ increases, the $+6W$ construction-based algorithm outperforms the $+2W$ algorithm. The average experimental ratio of the $+6W$ construction based algorithm is smaller both in the linear and exponential terminal selection methods.

\begin{figure}[ht]
\begin{minipage}{\textwidth}
    \centering
    \begin{subfigure}[b]{0.30\textwidth}
        \includegraphics[width=\textwidth]{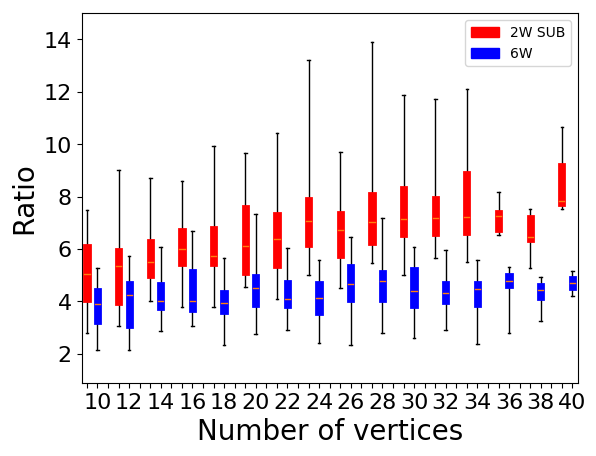}
    \end{subfigure}
    ~ 
    \begin{subfigure}[b]{0.30\textwidth}
        \includegraphics[width=\textwidth]{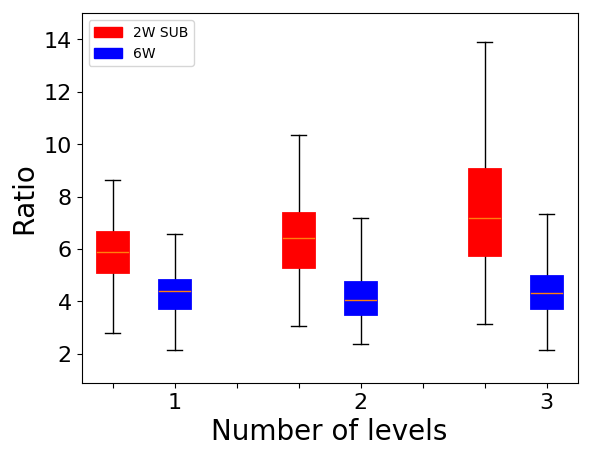}
    \end{subfigure}
    ~
    \begin{subfigure}[b]{0.30\textwidth}
        \includegraphics[width=\textwidth]{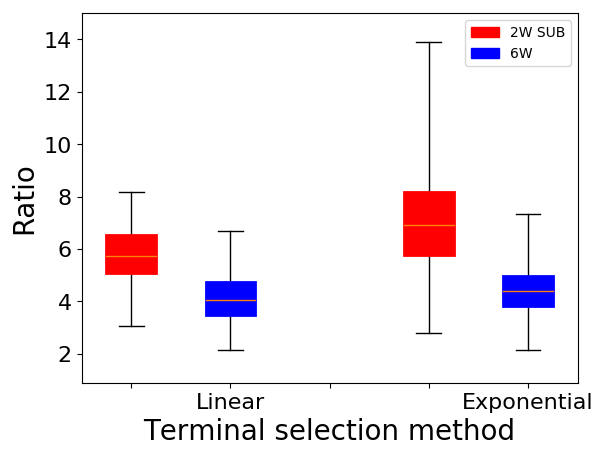}
    \end{subfigure}
    \caption{Performance of the pairwise $+2W$ and $+6W$ construction-based algorithms on Erd\H{o}s--R{\'e}nyi graphs w.r.t.\
      $n$, $\ell$, and terminal selection method.}
    \label{BoxPlots_ER_GLOB_GLOB}
\end{minipage}
\end{figure}

We describe the experimental results on random geometric graphs w.r.t. $n$, $\ell$, and the terminal selection method in Figure~\ref{BoxPlots_GE_GLOB_GLOB}. 
We can see that as $n$ gets larger the ratio of $+6W$ gets smaller. The situation is similar when $\ell$ increases.

\begin{figure}[H]

\begin{minipage}{\textwidth}
    \centering
    \begin{subfigure}[b]{0.30\textwidth}
        \includegraphics[width=\textwidth]{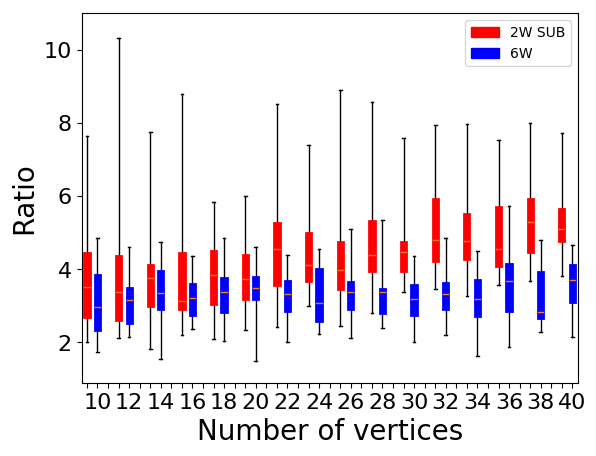}
    \end{subfigure}
    ~ 
    \begin{subfigure}[b]{0.30\textwidth}
        \includegraphics[width=\textwidth]{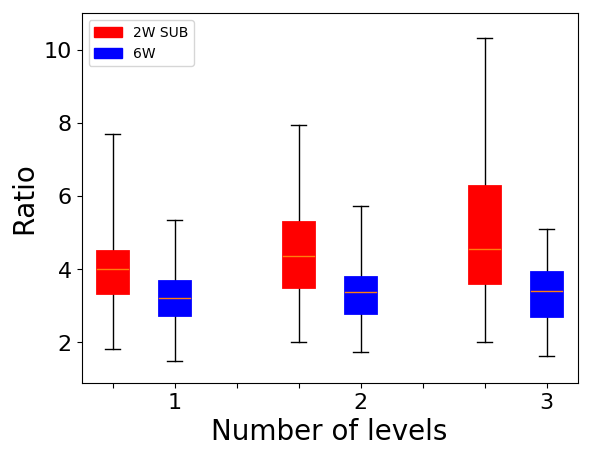}
    \end{subfigure}
    ~
    \begin{subfigure}[b]{0.30\textwidth}
        \includegraphics[width=\textwidth]{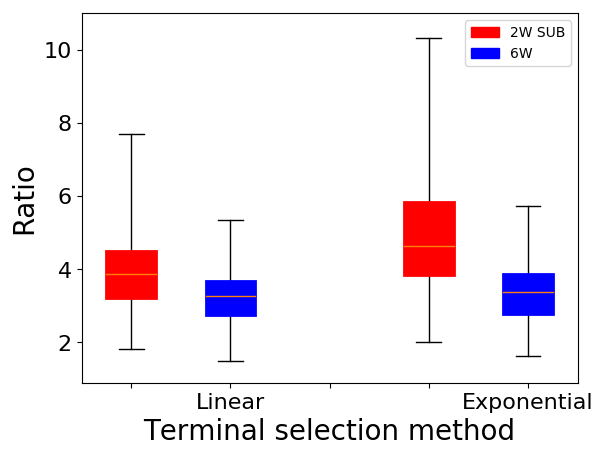}
    \end{subfigure}
    \caption{Performance of the pairwise $+2W$ and $+6W$ construction-based algorithms on random geometric graphs w.r.t.\
      $n$, $\ell$, and terminal selection method.}
    \label{BoxPlots_GE_GLOB_GLOB}
\end{minipage}
\end{figure}

\subsubsection{Experiment on Large Graphs}
\label{sec:large_graphs}
We generate some large instances on up to 500 vertices and run different multi-level spanner algorithms on them. We use $n = \{50, 100, 150, \ldots, 500\}$ and $\ell = \{1, 2, 3, \ldots, 10\}$. We describe the experimental results on Erd\H{o}s--R{\'e}nyi graphs w.r.t. $n$, $\ell$, and the terminal selection method in Figure~\ref{BoxPlots_ER_large}. We are comparing four multi-level algorithms, namely those using the $+2W$ subsetwise and $+2W(\cdot, \cdot)$, $+4W(\cdot, \cdot)$, $+6W$ pairwise constructions~\cite{ahmed2020weighted} as subroutines with $P=S \times S$. Since computing the optimal solution exactly via ILP is computationally expensive on large instances, we report the ratio in terms of relative sparsity, defined as the sparsity of the multi-level spanner returned by one algorithm divided by the minimum sparsity over the spanners returned by all four. The ratio of the $+6W$ construction based algorithm is lowest and the $+2W$ construction based algorithm is highest. This is expected since a higher additive error generally reduces the number of edges needed. Overall the ratio decreases as $n$ increases. This is because the significance of small additive error reduces as the graph size and distances get larger. The relative ratio for the $+2W$ construction increases as $\ell$ increases, and for the exponential terminal selection method.

\begin{figure}[ht]
\begin{minipage}{\textwidth}
    \centering
    \begin{subfigure}[b]{0.30\textwidth}
        \includegraphics[width=\textwidth]{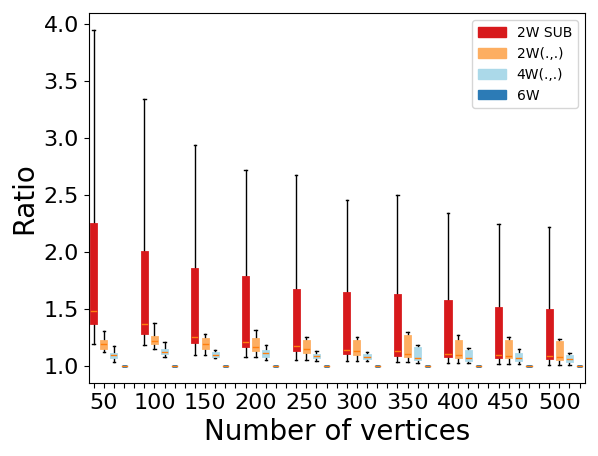}
    \end{subfigure}
    ~ 
    \begin{subfigure}[b]{0.30\textwidth}
        \includegraphics[width=\textwidth]{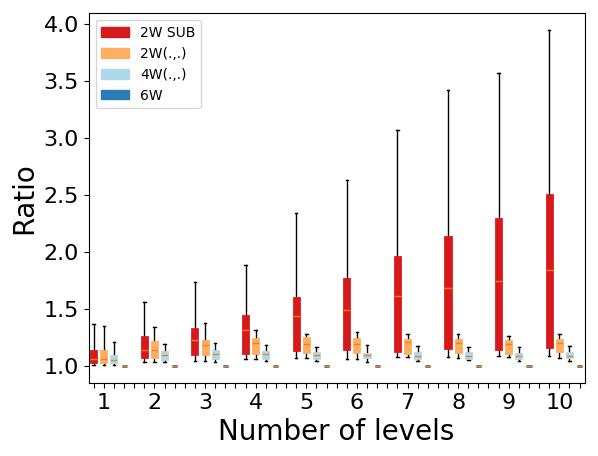}
    \end{subfigure}
    ~
    \begin{subfigure}[b]{0.30\textwidth}
        \includegraphics[width=\textwidth]{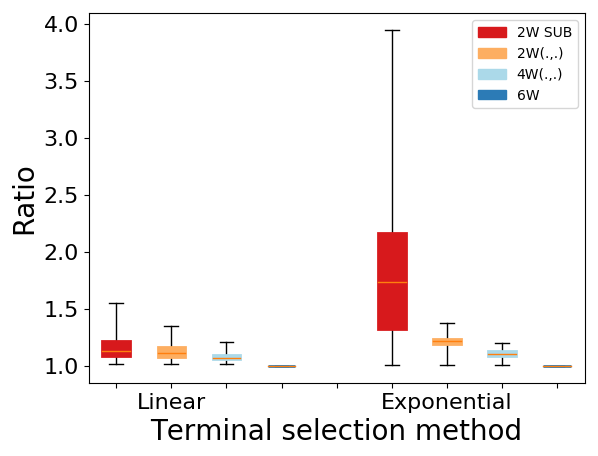}
    \end{subfigure}
    \caption{Performance of different approximation algorithms on large Erd\H{o}s--R{\'e}nyi graphs w.r.t.\
      $n$, $\ell$, and terminal selection method.}
    \label{BoxPlots_ER_large}
\end{minipage}
\end{figure}

We describe the experimental results on random geometric graphs w.r.t. $n$, $\ell$, and the terminal selection method in Figure~\ref{BoxPlots_GE_large}.

\begin{figure}[H]

\begin{minipage}{\textwidth}
    \centering
    \begin{subfigure}[b]{0.30\textwidth}
        \includegraphics[width=\textwidth]{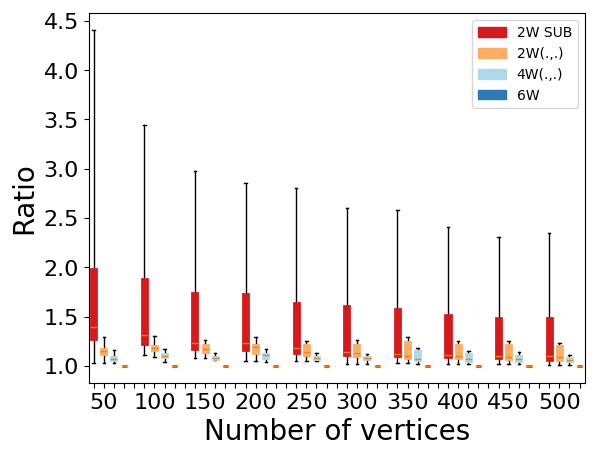}
    \end{subfigure}
    ~ 
    \begin{subfigure}[b]{0.30\textwidth}
        \includegraphics[width=\textwidth]{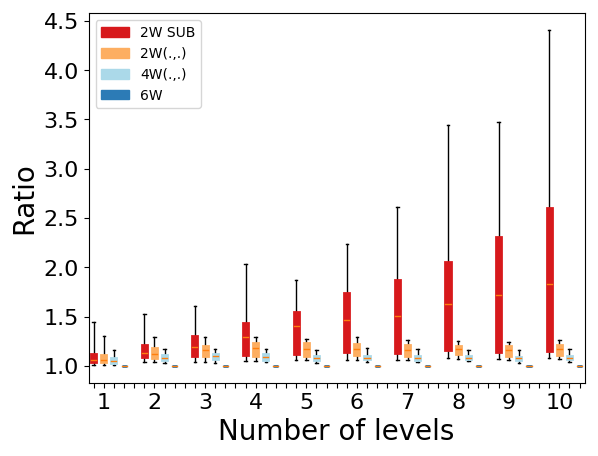}
    \end{subfigure}
    ~
    \begin{subfigure}[b]{0.30\textwidth}
        \includegraphics[width=\textwidth]{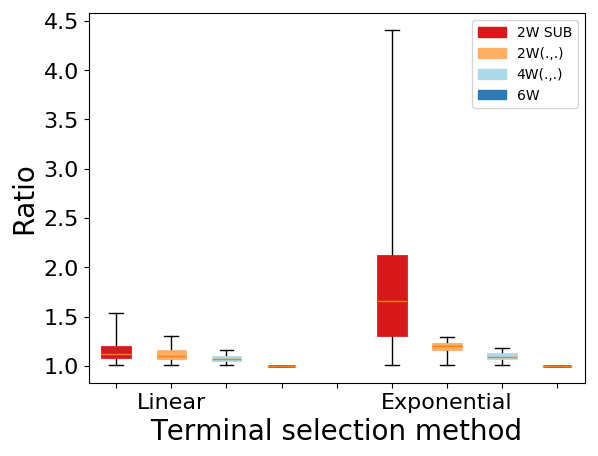}
    \end{subfigure}
    \caption{Performance of different approximation algorithms on large random geometric graphs w.r.t.\
      $n$, $\ell$, and terminal selection method.}
    \label{BoxPlots_GE_large}
\end{minipage}
\end{figure}

\subsubsection{Impact of Initial Parameters}
It is worth mentioning that the $+2$ subsetwise spanner~\cite{Cygan13} and $+2W$ subsetwise spanner (Section~\ref{sec:subsetwise}) begin with a clustering phase, while the algorithms described in Appendix~\ref{section:pairwise} begin with a $d$-light initialization. In $d$-light initialization, we add the $d$ lightest edges incident to each vertex; these edges tend to be on shortest paths. In practice, there may be relatively few edges which appear on shortest paths and some of these edges might be redundant. Hence, we compute $+2W(\cdot, \cdot)$ spanners with different values of $d$. We describe the experimental results on Erd\H{o}s--R{\'e}nyi graphs w.r.t. $n$, $\ell$, and the terminal selection method in Figure~\ref{BoxPlots_ER_d}. We have computed the ratio as described in Section~\ref{sec:large_graphs}. From the figures, we see that as we reduce the value of $d$ exponentially, the ratio decreases: in particular, the optimal choice of parameter $d$ in practice might be significantly smaller than the optimal value of $d$ in theory.
Generally, it could make sense in practical implementations of spanner algorithms to try all values $\{d, d/2, d/4, d/8, \dots\}$, computing $\log d$ different spanners, and then use only the sparsest one.
This costs only a $\log d$ factor in the running time of the algorithm, which is typically reasonable.

\begin{figure}[ht]
\begin{minipage}{\textwidth}
    \centering
    \begin{subfigure}[b]{0.30\textwidth}
        \includegraphics[width=\textwidth]{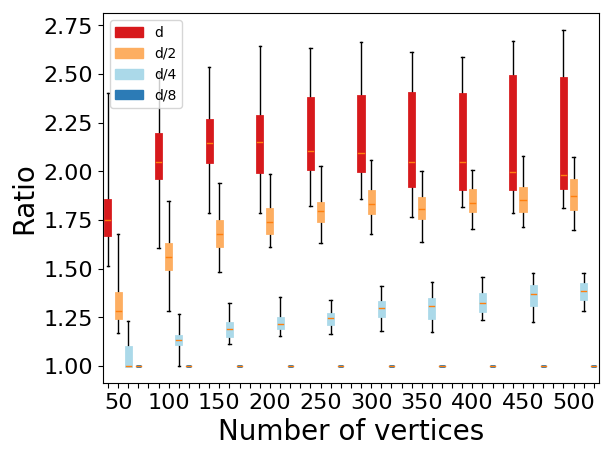}
    \end{subfigure}
    ~ 
    \begin{subfigure}[b]{0.30\textwidth}
        \includegraphics[width=\textwidth]{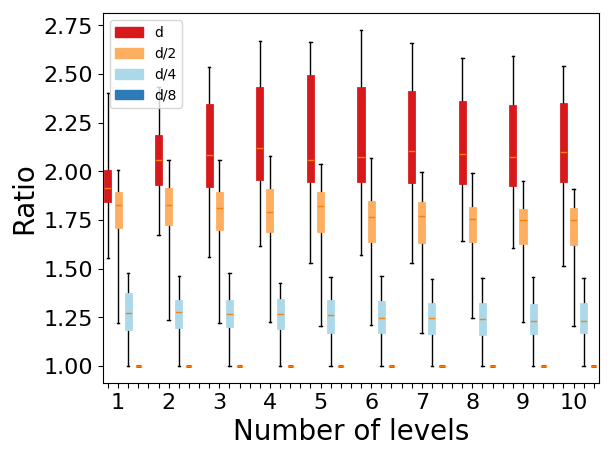}
    \end{subfigure}
    ~
    \begin{subfigure}[b]{0.30\textwidth}
        \includegraphics[width=\textwidth]{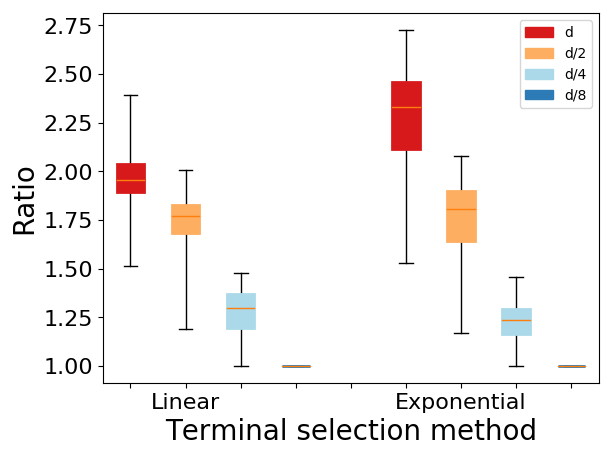}
    \end{subfigure}
    \caption{Impact of different values of $d$ on large Erd\H{o}s--R{\'e}nyi graphs w.r.t.\
      $n$, $\ell$, and terminal selection method.}
    \label{BoxPlots_ER_d}
\end{minipage}
\end{figure}

We describe the experimental results on random geometric graphs w.r.t. $n$, $\ell$, and the terminal selection method in Figure~\ref{BoxPlots_GE_d}. 
Again, the experiment suggests that we can exponentially reduce the value of $d$ and take the solution that has a minimum number of edges.

\begin{figure}[H]

\begin{minipage}{\textwidth}
    \centering
    \begin{subfigure}[b]{0.30\textwidth}
        \includegraphics[width=\textwidth]{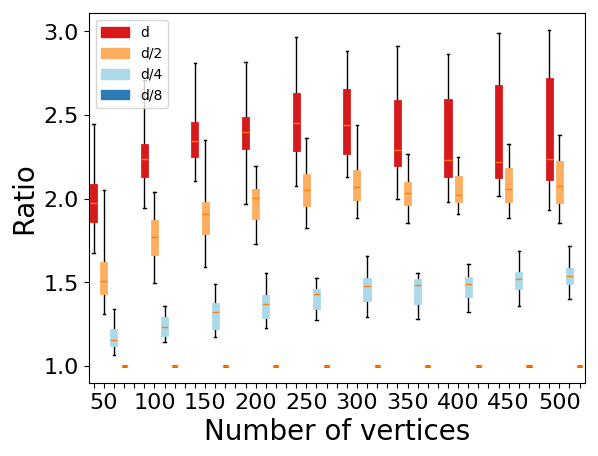}
    \end{subfigure}
    ~ 
    \begin{subfigure}[b]{0.30\textwidth}
        \includegraphics[width=\textwidth]{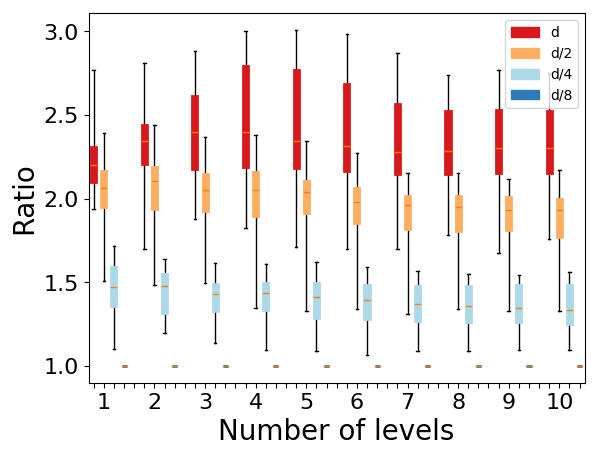}
    \end{subfigure}
    ~
    \begin{subfigure}[b]{0.30\textwidth}
        \includegraphics[width=\textwidth]{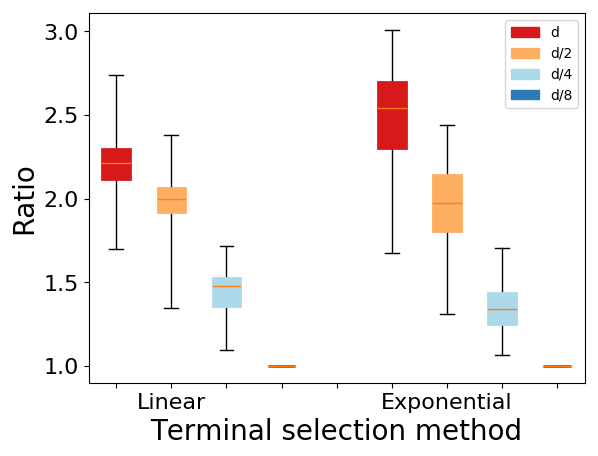}
    \end{subfigure}
    \caption{Impact of different values of $d$ on large random geometric graphs w.r.t.\
      $n$, $\ell$, and terminal selection method.}
    \label{BoxPlots_GE_d}
\end{minipage}
\end{figure}

\subsection{Running Time}

We now provide the running times of the different algorithms. 
We show the running time of the ILP on Erd\H{o}s--R{\'e}nyi graphs w.r.t. $n$, $\ell$, and terminal selection method in Figure~\ref{BoxPlots_ER_time_all}. The running time of the ILP increases exponentially as $n$ increases, as expected. The execution time of a single level instance with 45 vertices is more than 64 hours using a 28 core processor. Hence, we kept the number of vertices less than or equal to 40 for our small graphs. The experimental running time should increase as $\ell$ increases, but we do not see that pattern in these plots because some of the instances were not able to finish in four hours. 

\begin{figure}[ht]
\begin{minipage}{\textwidth}
    \centering
    \begin{subfigure}[b]{0.30\textwidth}
        \includegraphics[width=\textwidth]{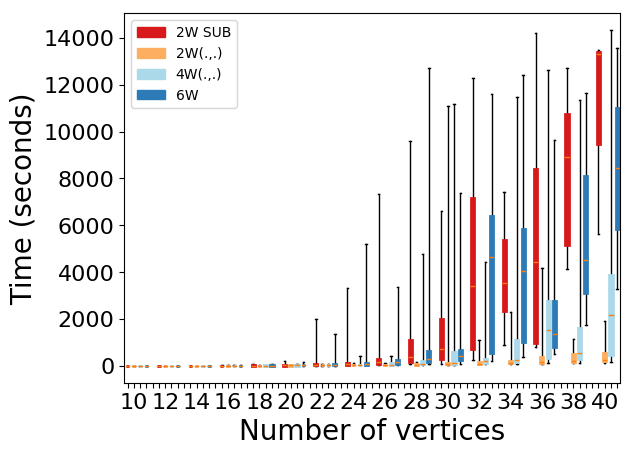}
    \end{subfigure}
    ~ 
    \begin{subfigure}[b]{0.30\textwidth}
        \includegraphics[width=\textwidth]{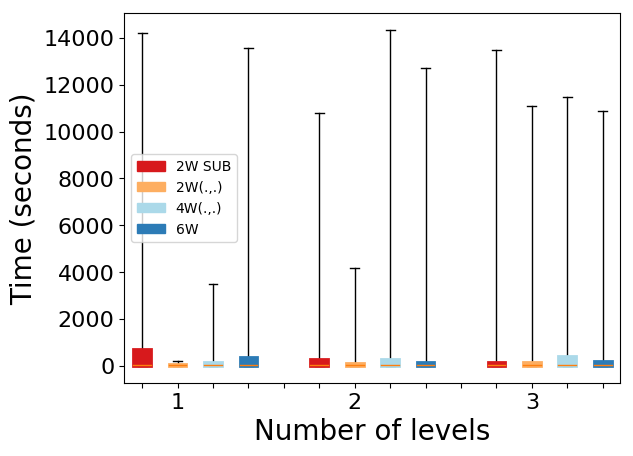}
    \end{subfigure}
    ~
    \begin{subfigure}[b]{0.30\textwidth}
        \includegraphics[width=\textwidth]{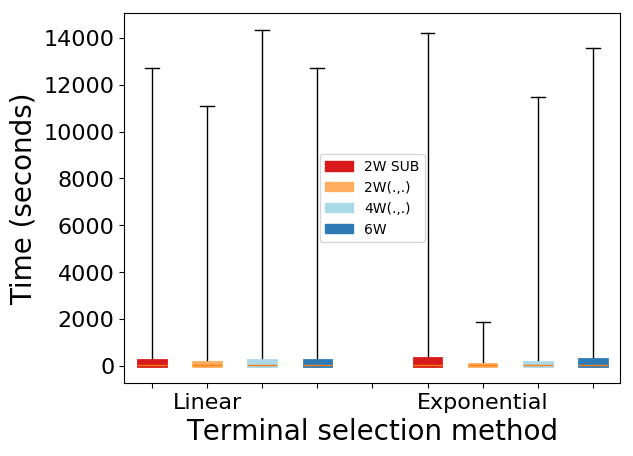}
    \end{subfigure}
    \caption{Running time of all exact algorithms on Erd\H{o}s--R{\'e}nyi graphs w.r.t.\
      $n$, $\ell$, and terminal selection method.}
    \label{BoxPlots_ER_time_all}
\end{minipage}
\end{figure}

We provide the experimental running time of the approximation algorithm on Erd\H{o}s--R{\'e}nyi graphs in Figure~\ref{BoxPlots_ER_time_apprx_large}. The running time of the $+2W$ construction-based algorithm is the largest. Overall, the running time increases as $n$ increases. There is no straightforward relation between the running time and $\ell$. Although the number of calls to the single level subroutine increases as $\ell$ increases, it also depends on the size of the subset in a single level. Hence, if the subset sizes are larger, then it may take longer for small $\ell$. The running time of the linear method is larger.

\begin{figure}[ht]
\begin{minipage}{\textwidth}
    \centering
    \begin{subfigure}[b]{0.30\textwidth}
        \includegraphics[width=\textwidth]{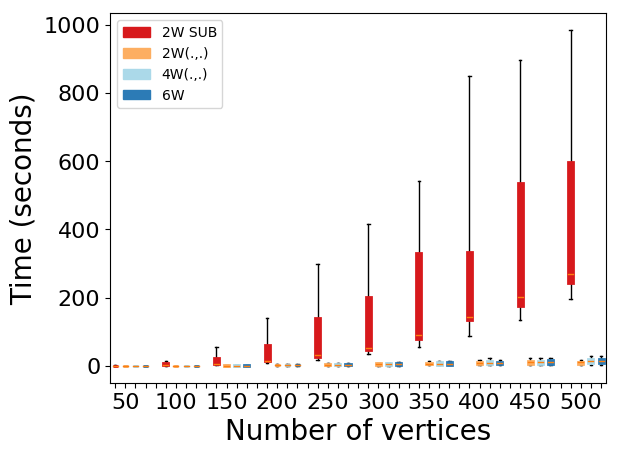}
    \end{subfigure}
    ~ 
    \begin{subfigure}[b]{0.30\textwidth}
        \includegraphics[width=\textwidth]{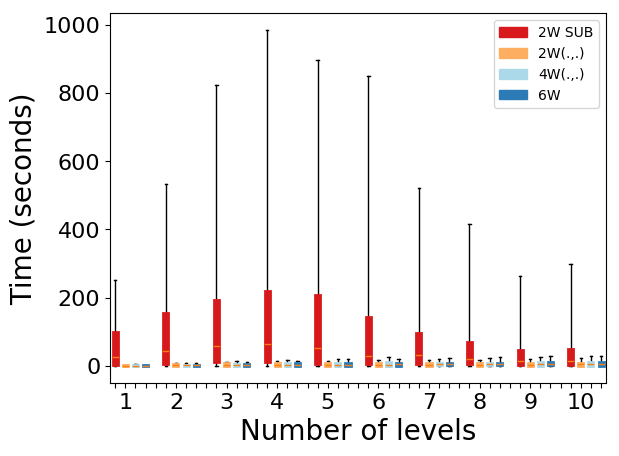}
    \end{subfigure}
    ~
    \begin{subfigure}[b]{0.30\textwidth}
        \includegraphics[width=\textwidth]{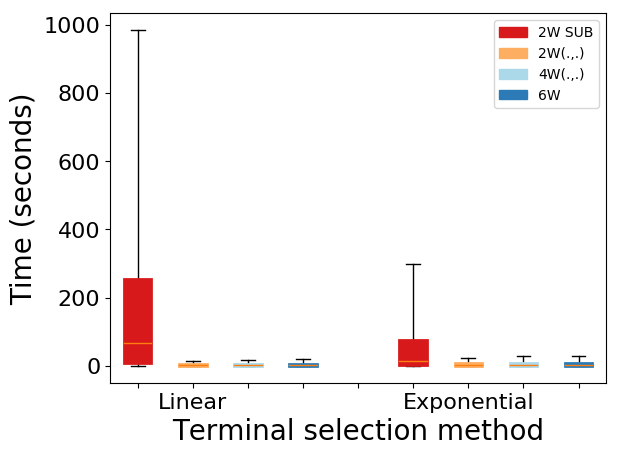}
    \end{subfigure}
    \caption{Running time of all approximation algorithms on large Erd\H{o}s--R{\'e}nyi graphs w.r.t.\
      $n$, $\ell$, and terminal selection method.}
    \label{BoxPlots_ER_time_apprx_large}
\end{minipage}
\end{figure}

The running times appear reasonable in other settings too; see the supplemental Github repository for details and experimental results.

\section{Conclusion}

We have provided a framework where we can use different spanner subroutines to compute multi-level spanners of varying additive error. We additionally introduced a generalization of the $+2$ subsetwise spanner~\cite{Cygan13} to integer edge weights, and illustrate that this can reduce the $+8W$ error in~\cite{ahmed2020weighted} to $+6W$. A natural question is to provide an approximation algorithm that can handle different additive error for different levels. We also provided an ILP to find the exact solution for both global and local settings. Using the ILP and the given spanner constructions, we have run experiments and provided the experimental approximation ratio over the graphs we generated. 
The experimental results suggest that the $+2W$ clustering-based approach is slower than the initialization based approaches. We provided a method of changing the initialization parameter $d$ which reduces the sparsity in practice.

\bibliography{references}
\newpage

\appendix

\end{document}